\documentclass[12pt]{article}

\usepackage{geometry}
\usepackage{doc}

\usepackage{url}

\usepackage{tikz-cd}
\usetikzlibrary{matrix,arrows,decorations.pathmorphing}

\usepackage{graphicx}
\usepackage{epstopdf}
\usepackage{amsmath,amsfonts,amssymb,amsthm,amscd}
\usepackage[mathcal]{euscript} 
\usepackage{bbm} 
\usepackage[T1]{fontenc}
\usepackage{pdfsync}
\usepackage{pifont}
\usepackage[adobe-utopia]{mathdesign}
\input xy
\xyoption{all}

\theoremstyle{plain}
\newtheorem{thm}{Theorem}[section]
\newtheorem{prop}[thm]{Proposition}

\theoremstyle{definition}

\theoremstyle{remark}

\newtheorem*{rmk}{Remark}

\newtheorem*{ex}{Example}

\newtheorem*{example}{Example}

\newcommand{\beq}{\begin{equation*}}
\newcommand{\eeq}{\end{equation*}}

\newcommand{\ket}[1]{|{#1}\rangle}
\newcommand{\bra}[1]{\langle #1 |}

\title{Geometric formulation of quantum mechanics}
\author{Hoshang Heydari}
\date{}

\begin{document}

\maketitle

\abstract{
Quantum mechanics is among the most important and successful mathematical model for describing our physical reality. The traditional formulation of quantum mechanics is linear and algebraic. In contrast classical mechanics is a geometrical and non-linear theory that is defined on a symplectic manifold.
However, after invention of general relativity, we are convinced that geometry is physical and effect us in all scale.
Hence the geometric formulation of quantum mechanics sought to give a unified picture of physical systems based on its underling geometrical structures, e.g., now, the states are represented by points of a symplectic manifold with a compatible Riemannian metric, the observables are real-valued functions on the manifold, and the quantum evolution is governed by a symplectic flow that is generated by a Hamiltonian function. In this work we will give a compact introduction to main ideas of geometric formulation of quantum mechanics. We will provide the reader  with the details of geometrical structures of both pure and mixed quantum states. We will also discuss and review some important applications of geometric  quantum mechanics.
}


\tableofcontents




%
%
\section{Introduction}
\label{introduction}

In the geometrical description of  classical mechanics  the states are represented by the points of a
symplectic manifold $\mathcal{M}$ which is called the phase space \cite{Abraham_1978}. The space of observables consists of the
real-valued and smooth functions on the phase space. The measurement of an observable $f:\mathcal{M}\longrightarrow \mathbb{R}$ in a
state $p\in \mathcal{M}$ is given by $f(p)$. The space of observables is  equipped
with the structure of a commutative and associative algebra. The symplectic structure of the phase space also provides it with
the Poisson bracket. An observable $f$ is associated with a vector field $X_{f}$. Hence, flow on the phase space is generated by each observable. Moreover, the dynamics
is given by a particular observer called the Hamiltonian $H$ and the flow generated by the Hamiltonian vector field $X_{H}$
 describes  the time evolution of the system on the phase space.

In quantum mechanics
the systems correspond to rays in the Hilbert space $\mathcal{H}$, and the observables are represented by hermitian/self-adjoint
linear operators on $\mathcal{H}$. Moreover, the space of observables is a real vector
space equipped with  two algebraic structures, namely  the Jordan product
and the commutator bracket. Thus  the space of observables is equipped with the
structure of a Lie algebra.
However, the measurement theory is different compare with the classical mechanics. In
 the standard interpretation of quantum mechanics, the measurement
of an observable $\hat{A}$ in a state $\ket{\psi}\in\mathcal{H}$ gives an eigenvalue of $\hat{A}$.
The observable $\hat{A}$ also gives rise to
a flow on the state space as in the classical theory. But the flow is generated by the 1-parameter group $e^{i\hat{A}t}$
that preserves the linearity of the Hilbert space. The  dynamics is governed by a specific observable, called the
Hamiltonian operator $\hat{H}$.

One can directly see that these theories both have  several points in common and
 also in difference. The classical mechanical framework is geometric and non-linear. But the quantum
mechanical framework is algebraic and linear. Moreover, the standard postulates of quantum mechanics
cannot be stated without reference to this linearity.
However, some researcher think that this difference seems quite surprising \cite{Gunter_1977,Kibble_1979,Ashtekar_etal1998,Brody_etal1999} and deeper investigation shows that quantum mechanics is not a linear theory either.
Since, the space of physical systems is not the Hilbert space $\mathcal{H}$ but it is the projective Hilbert space $\mathcal{P}(\mathcal{H})$ which  is a nonlinear manifold.
Moreover, the Hermitian inner-product of the Hilbert space naturally equips the projective space with
the structure of a K\"{a}hler manifold which
is also a symplectic manifold like the classical mechanical phase space $\mathcal{M}$. The projective space $\mathcal{P}(\mathcal{H})$ is usually called quantum phase space of the pure quantum states.

Let $\hat{H}$ be a Hamiltonian operator on $\mathcal{H}$. Then we can take  its expectation
value to obtain a real function on the Hilbert space which admits a
 projection $h$ to  $\mathcal{P}(\mathcal{H})$. The flow $X_{h}$ is exactly the
flow defined by the Schr\"{o}dinger equation on  $\mathcal{H}$. This means that Schr\"{o}dinger
evolution of quantum theory is the Hamiltonian flow on $\mathcal{P}(\mathcal{H})$.
These similarities show us that the classical mechanics and quantum mechanics have many points in common. However, the quantum phase space has additional structures such as a Riemannian metric  which are missing in the classical mechanics (actually Riemannian metric exists but it is not important in the classical mechanics). The Riemannian metric is part of underlying K\"{a}hler structure of quantum phase space. Some important features such as uncertainty relation and state vector reduction in quantum measurement processes are provided by the Riemannian metric. \\
In this work we will also illustrate the interplay between theory and the applications of geometric formulation of quantum mechanics. Recently, many researcher   \cite{Adler2000,Anandan1990b,Anandan1991,Gibbson1992,Marsden1999,Schilling1996, Marmo1,Marmo2,Montgomery1991,Hosh2,Levay} have contributed to development  of geometric formulation of quantum mechanics and how this formulation provide us with insightful information about our quantum world with many applications in foundations of quantum mechanics and quantum information theory such as quantum probability, quantum uncertainty relation, geometric phases, and quantum speed limit.
\\
In the early works, the most effort in geometric quantum mechanics were concentrated around understanding geometrical structures of pure quantum states and less attention were given to the mixed quantum states. Uhlmann was among  the first researcher to consider a geometric formulation of  mixed quantum states with the emphasizes on geometric phases \cite{Uhlmann1986,Uhlmann(1989),Uhlmann(1991)}. Recent attempt to uncover hidden geometrical structures of mixed quantum states were achieved in the following works \cite{GP,MB,DD,GQE,GUR,QSL}.
Some researcher also argue that geometric formulation of quantum mechanics could lead to a generalization of quantum mechanics \cite{Ashtekar_etal1998}. However, we will not discuss such a generalization in this work. Instead we concentrate our efforts to give an introduction to its basic structures with some applications.
In particular,  in section \ref{sec1} we review some important mathematical tools such as Hamiltonian dynamics,  principal fiber bundles, and momentum map. In section \ref{sec2} we will discuss the basic structures of geometric quantum mechanics including quantum phase space, quantum dynamics, geometric uncertainty relation,
 quantum measurement,  geometric postulates of quantum mechanics, and  geometric phase for pure quantum states. In section \ref{sec3} we will extend our discussion to more general quantum states, namely the mixed quantum states represented by density operators. Our review on the geometric quantum mechanics of mixed quantum states includes purification, symplectic reduction, symplectic and Riemannian structures, quantum energy dispersion,  geometric uncertainty relation, geometric postulates of quantum mechanics, and geometric phase. Finally in section \ref{sec4} we give a conclusion and an outlook. Note that we assume that reader are familiar with basic topics of differential geometry.

\section{Mathematical structures}\label{sec1}
Mathematical structures are important in both classical and  quantum physics. In algebraic description of  quantum mechanics linear algebra and operator algebra are the most preferred structures for describing physical systems. However, in geometric quantum mechanics the most important mathematical structures are geometrical such as
 Hamiltonian dynamics,  principal fiber bundles, and momentum maps. In this section we will give a short introduction to these topics.
\subsection{Hamiltonian dynamics }
In Hamiltonian mechanics the space of states or phase space is a differential manifold $\mathcal{M}$ equipped with a symplectic form  $\omega$ which plays an important role in describing the time evolution of the states of the system.
Here we will give a short introduction to Hamiltonian dynamics. For a detail discussion of Hamiltonian dynamics we recommend the following classical book \cite{Abraham_1978}.\\
Let $\mathcal{M}$ be a smooth manifold with $p\in \mathcal{M}$ and $\mathcal{T}(\mathcal{M})$ be the tangent space of $\mathcal{M}$.
Moreover, let
\begin{equation}
\omega: \mathcal{T}_{p}(\mathcal{M})\times \mathcal{T}_{p}(\mathcal{M}) \longrightarrow \mathbb{R}
\end{equation}
 be a two-form on $\mathcal{M}$. Then $\omega$ is called symplectic if
\begin{enumerate}
  \item $\omega$ is closed, $d\omega=0$, and
  \item $\omega$ is non-degenerated, that is, $\omega(u,v)=0$ for all $u \in\mathcal{T}_{p}(\mathcal{M})$ whenever $v=0$.
\end{enumerate}
The pair $(\mathcal{M},\omega)$ are called a symplectic manifold.
If $H:\mathcal{M}\longrightarrow \mathbb{R}$ is a smooth function on $\mathcal{M}$, then $dH$ is a 1-form on $\mathcal{M}$. Moreover, let $X:\mathcal{M}\longrightarrow \mathcal{T}(\mathcal{M})$ be a vector field. Then we define
a contraction map $\imath :  \mathcal{T}(\mathcal{M}) \longrightarrow \mathcal{T}^{*}(\mathcal{M})$ by  $\imath_{X_{H}}\omega=\omega (X_{H},\cdot)$. The vector field $X$ is called symplectic if  $\imath_{X}\omega$ is  closed. Furthermore, a vector field $X_{H}:\mathcal{M}\longrightarrow \mathcal{T}(\mathcal{M})$ is called a Hamiltonian vector field  with a Hamiltonian function $H$ if it satisfies
\begin{equation}
\imath_{X_{H}}\omega=\omega (X_{H},\cdot)= d H.
\end{equation}
A Hamiltonian system  consists of the following triple $(\mathcal{M},\omega, X_{H} )$.
Let $X_{H}\in \mathcal{T}(\mathcal{M})$ be a Hamiltonian vector field. Then $X_{H}$ generates the one-parameter group of diffeomorphism $\mathrm{Diff}(\mathcal{M})$
\begin{equation}
\{\phi^{H}_{t}\}:\mathbb{R}\times \mathcal{M}\longrightarrow \mathcal{M},
\end{equation}
 where $\phi^{H}_{t}\in\mathrm{Diff}(\mathcal{M})$ satisfies
 \begin{itemize}
   \item $\frac{d}{dt}\phi^{H}_{t}=X_{H}\circ \phi^{H}_{t}$ with $\phi^{H}_{0}=\mathrm{id}$, and
   \item $\phi^{H}_{t+s}=\phi^{H}_{t}(\phi^{H}_{s}(x))$ for all $t,s\in R$ and $x\in \mathcal{M}$.
 \end{itemize}
 For a Hamiltonian system $(\mathcal{M},\omega, X_{H} )$  each point $x\in \mathcal{M}$ corresponds to a state of system and the symplectic manifold $\mathcal{M}$ is called the state space or the phase space of the system. In such a classical system, the observables are real-valued functions on the phase space. Let $K:\mathcal{M}\longrightarrow \mathcal{M}$ be a function. Then $K$  is constant along the orbits of the flow of the Hamiltonian vector field $X_{H}$ if and only if the Poisson bracket defined by
 \begin{equation}
\{K,H\}_{\zeta}=\omega_{\zeta}(X_{K},X_{H})=dK(X_{H})(\zeta)
\end{equation}
vanishes for all $\zeta\in \mathcal{M}$. Assume $X_{H}$ is a Hamiltonian vector on $\mathcal{M}$, and let  $x$ be a point of $\mathcal{M}$. Moreover, let $\phi_{t}$ be one-parameter group generated by $X_{H}$ in a neighborhood of the point $x$. If we assume that the initial state is $x$, then the evolution of the state  can be described by the map
$\zeta_{x}:\mathbb{R }\longrightarrow\mathcal{M}$  defined by $\zeta_{x}(t)=\phi_{t}(x)$ with initial state $\zeta_{x}(0)=x$. Under these assumptions the trajectory of $\zeta_{x}(t)$ is determined by the Hamilton's equations
 \begin{equation}\label{HE}
 \zeta^{'}=X_{H}(\zeta).
\end{equation}
Note that \begin{equation}
 \mathcal{L}_{ X_{H}}\omega = \imath_{X_{H}}d\omega+d(\imath_{X_{H}}\omega)=0+ddH=0,
 \end{equation}
 where $\mathcal{L}_{ X_{H}}$ is the Lie derivative, implies that the flow $\phi_{t}$ preserves the symplectic  structure, that is $\phi^{*}_{t} \omega=\omega$. If the phase space $\mathcal{M}$ is compact, then $\zeta_{x}(t)$ is an integral curve of the Hamiltonian vector field $X_{H}$ at the point $x$.
\begin{thm}
Consider a Hamiltonian system $(\mathcal{M},\omega, X_{H} )$. If $\zeta(t)$ is an integral curve of $X_{H}$, then energy function $H(\zeta(t))$ is constant for all $t$. Moreover, the flow $\phi_t$  of $H$ satisfies  $H\circ\phi_t=H$.
 \end{thm}
\begin{proof}
By the Hamilton's equation (\ref{HE}) we have
\begin{eqnarray}
  \frac{d H(\zeta(t))}{dt} &=& d H_{\zeta(t)}(\zeta^{'}(t))
= d H_{\zeta(t)}((X_{H})_{\zeta(t)}))
  \\\nonumber &=& \omega((X_{H})_{\zeta(t)},(X_{H})_{\zeta(t)})=0.
\end{eqnarray}
Thus we have shown that $H(\zeta(t))$ is constant for all $t$.
 \end{proof}
 Let $\mathcal{M}$ be a manifold. Then an almost complex structure on $\mathcal{M}$ is an automorphism of its tangent bundle $J:\mathcal{T}\mathcal{M}\longrightarrow \mathcal{T}\mathcal{M}$ that satisfies $J^{2}=-\mathbf{1}$.
Moreover, the almost complex structure is  a complex structure if it is integrable, meaning that a rank two tensor, usually called the Nijenhuis tensor vanishes \cite{DaSilva}.
\\
Let  $(\mathcal{M},\omega)$ be symplectic manifold. Then a K\"{a}hler  manifold is symplectic manifold equipped with an integrable compatible complex structure. Moreover, $(\mathcal{M},\omega)$ being a K\"{a}hler  manifold implies that $\mathcal{M}$ is a complex manifold. Thus a K\"{a}hler form is a closed, real-valued, non-degenerated $2$-form compatible with the complex structure.
\begin{example}
Let $(q^{1},q^{2},\ldots,q^{n},p_{1},p_{2},\ldots, p_{n})$ be canonical coordinate for the symplectic form, that is $\omega=\sum^{n}_{i=1}q^{i}\wedge p_{i}$. Then in these coordinates we have
\begin{equation}
X_{H}=\left(\frac{\partial H}{\partial p_{i}},-\frac{\partial H}{\partial q^{i}}\right)=J\cdot d H,~~ ~~ J=\left(
           \begin{array}{cc}
             0 & I_{n} \\
             -I_{n} & 0 \\
           \end{array}
         \right),
\end{equation}
where $I_{n}$ is a $n\times n$ identity matrix.
\end{example}
 In the following sections we will show that the quantum  dynamics governed by Schr\"{o}dinger and von Neumann equations  can be described by Hamiltonian dynamics outlined in this section.

\subsection{Principal fiber bundle}
One important mathematical tool used in the geometric formulation of quantum physics is  principal fiber bundles. In this section we will introduce the reader to the basic definition and properties of principal fiber bundles and in the following sections we will apply the tool to the quantum theory.\\
Let $\mathcal{S}$ and $\mathcal{P}$ be differentiable manifolds and $G$ be a Lie group. Then a differentiable principal fiber bundle
\begin{equation}\label{pfb}
\xymatrix{
G \ar@{^{(}->}[r] & \mathcal{S}\ar[r]^{\pi} &\mathcal{P}}
\end{equation}
consists of the total space $\mathcal{S}$  and the action of $G$ on $\mathcal{S}$  that satisfies\begin{itemize}
\item $G$ acts freely from the right on $\mathcal{S}$, that is $\eta: \mathcal{S}\times G\longrightarrow \mathcal{S}$ defined by $\eta(p,g)=p\cdot g$.
\item the base space $\mathcal{P}=\mathcal{S}/G$ is a quotient space with $\pi$  being differentiable submersion.
 \item Each $\varrho\in \mathcal{P}$  has an open neighborhood $U$ and a diffeomorphism $\varphi:\pi^{-1}(U)$ $ \longrightarrow U\times G$ such that
     $\varphi(p)=(\pi(p),\phi(p))$ whit $\phi:\pi^{-1}(U) \longrightarrow  G$  defined by $\phi(p\cdot g)=\phi(p)g$, for all $p\in \pi^{-1}(U)$ and $g\in G$.
 \end{itemize}
If $p\in \pi^{-1}(\varrho)$, then $\pi^{-1}(\varrho)$ is the set of points $p\cdot g$, for all $g\in G$.
\begin{ex}\label{hs}
One important example of such a principal fiber bundle is the construction of homogeneous spaces  \begin{equation}
\xymatrix{
H \ar@{^{(}->}[r] & G\ar[r]^{\pi} &G/H},
\end{equation}
where  $H$ is a closed subgroup of $G$ and $G/H=\{g H:g\in G\}$ is the set of all left coset of $H$ in $G$.
\end{ex}

\begin{figure}[t]
\centering
\includegraphics[scale=.65]{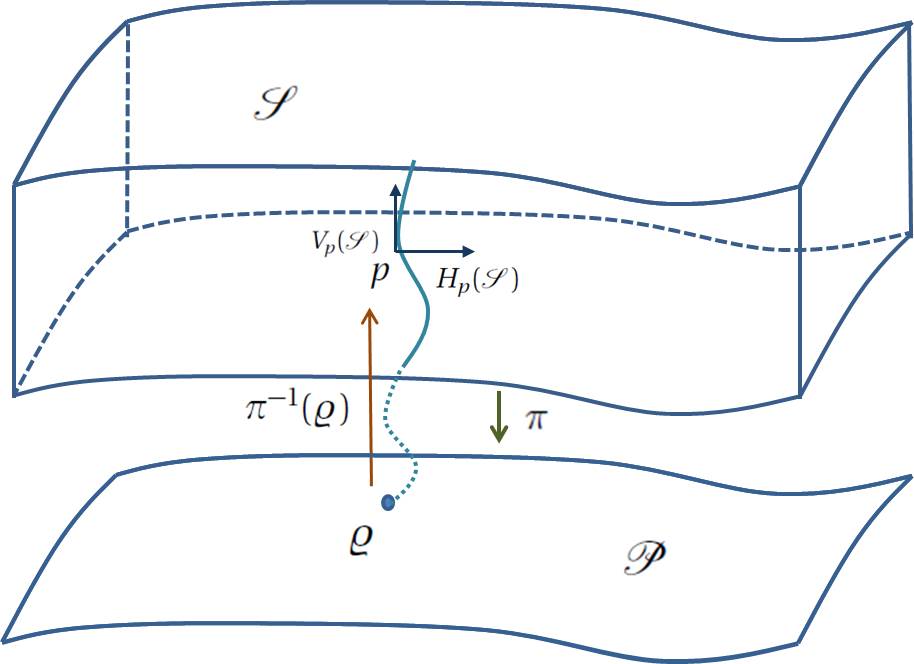}
%
%
\caption{Illustration of the bundle $\pi$ and the decomposition of $\mathcal{T}_{p}(\mathcal{S})$.}
\label{fig11}      
\end{figure}
Now, if we consider the principal fiber bundle $\xymatrix{
G \ar@{^{(}->}[r] & \mathcal{S}\ar[r]^{\pi} &\mathcal{P}}$, then the tangent space of $\mathcal{S}$ can be decompose as
\begin{equation}
\mathcal{T}_{p}(\mathcal{S})=V_{p}(\mathcal{S})\oplus H_{p}(\mathcal{S}),
\end{equation}
where $V_{p}=\{v\in \mathcal{T}_{p}(\mathcal{S}):\mathcal{T}_{p}\pi(v)=0\}$ is called a vertical subspace and $H_{p}(\mathcal{S})$ is called the horizontal subspace. Note that the horizontal subspace is transverse to the vertical subspace, see Figure 1. A  principal connection $\mathcal{A}$ in $\mathcal{S}$ is an assignment of the subspace $H_{p}(\mathcal{S})$ of $\mathcal{T}_{p}(\mathcal{S})$ such a that  $H_{pg}(S)=(R_{p})_{*}H_{p}(S)$ for each $p\in \mathcal{S}$ and $g\in G$ with $R_{g}:\mathcal{S}\longrightarrow \mathcal{S}$ defined by $R_{g}p=p\cdot g$.
A vector $v\in \mathcal{T}_{p}(\mathcal{S})$ is called horizontal if $v\in H_{p}(\mathcal{S})$ otherwise it is called vertical vector. A vertical vector is  denoted by  $\mathrm{vert}(v)$ and a horizontal one is  denoted by $\mathrm{hor}(v)$. Thus the vector $v$ can be written  as $v= \mathrm{vert}(v)+\mathrm{hor}(v)$. A curve $c(t)\in \mathcal{S}$ with $t\in\{0,1\}$ is called horizontal if $\frac{d c(t)}{dt}$, is horizontal. Now, let $ \gamma (t)\in \mathcal{P}$. Then $\tilde{\gamma}(t)$ is a lift of $\gamma(t)$
 if $\pi(\tilde{\gamma}(t))=\gamma(t)$. Moreover, if  $\tilde{\gamma}(t)$ is a horizontal curve of $\gamma$, then $\tilde{\gamma}(t)$  is called a horizontal lift. We will consider other important principal fiber bundles in the following sections.

\subsection{Momentum map}
The momentum map is also an import tool in the geometric formulation quantum mechanics,  specially in the construction of phase space of mixed quantum states. Here we give a very short introduction to the momentum map.

Let $\mathcal{S}$ be a symplectic manifold and $G$ be a Lie group acting on $\mathcal{S}$. Then the orbit of $G$ through $p\in \mathcal{S}$ is defined by
\begin{equation}
\mathcal{O}_{p}=\{\psi_{g}(p):g\in G\}
\end{equation}
and the stabilizer or isotropy subgroup of $G$ is defined by
\begin{equation}
G_{p}=\{g\in G: \psi_{g}(p)=p\}.
\end{equation}
The action of $G$ on $M$  is called transitive if there is only one orbit, it is called free if all stabilizers are trivial, and it is called locally free if all stabilizers are discrete.\\
Let $\mathcal{S}$ be a symplectic manifold and $G$ be a Lie group acting on $\mathcal{S}$ from the left.
Then the mapping
\begin{equation}\label{mm}
J:\mathcal{S}\longrightarrow \mathbb{g}^{*}
\end{equation}
is a momentum map, where $\mathbb{g}$ is the Lie algebra of $G$ and $\mathbb{g}^{*}$ is dual of $\mathbb{g}$. Moreover,  for a weakly regular value $\mu\in\mathbb{g}^{*}$ of $J$ the reduced space
\begin{equation}\mathcal{S}_{\mu}=J^{-1}(\mu)/G_{\mu}
\end{equation}
is a smooth manifold with the canonical projection being a surjective submersion, where
\begin{equation}
G_{\mu}=\{g\in G: \mathrm{Ad}^{*}_{g}\mu=\mu\}
\end{equation}
is the isotropy subgroup at $\mu$ for the co-adjoint action. The following theorem is called the symplectic reduction theorem \cite{Marsden_etal1974}.
\begin{figure}[t]
\centering
\includegraphics[scale=.65]{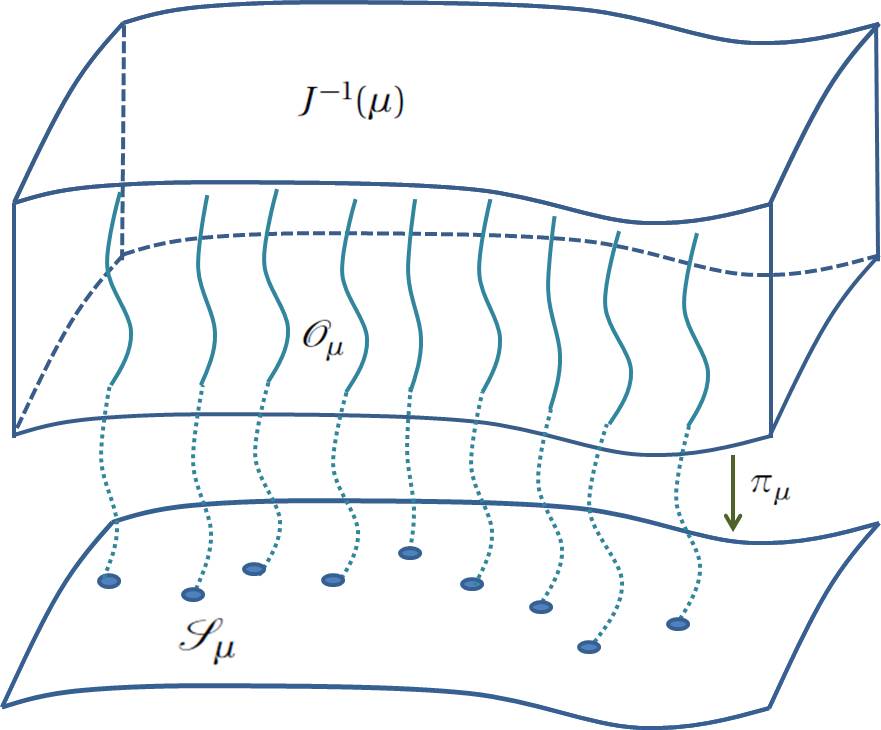}
%
%
\caption{Symplectic reduction, where $\mathcal{O}_{\mu}$ denotes the orbits of $G_{\mu}$.}
\label{fig11}      
\end{figure}
\begin{thm}\label{srt}
Consider a symplectic manifold $(\mathcal{S},\Omega)$ endow with a Hamiltonian left action of a Lie group $G$ and a momentum map $J:\mathcal{S}\longrightarrow \mathbb{g}^{*}$. Suppose that $\mu\in \mathbb{g}^{*}$ is a regular value of $J$ and the group $G_{\mu}$ acts freely and properly on $J^{-1}(\mu)$. Then the reduced phase space $\mathcal{S}_{\mu}=J^{-1}(\mu)/G_{\mu}$ has a unique symplectic form $\omega_{\mu}$ such that $\pi^{*}\omega_{\mu}=i^{*}_{\mu}\Omega$, where $i_{\mu}:J^{-1}(\mu) \longrightarrow \mathcal{S}$ is inclusion and $\pi_{\mu}:J^{-1}(\mu) \longrightarrow \mathcal{S}_{\mu}$ is a canonical projection.
\end{thm}
 If $H:\mathcal{S}\longrightarrow \mathbb{R}$ is a $G$-invariant Hamiltonian, then it induces  a Hamiltonian
$H_{\mu}:\mathcal{S}_{\mu}\longrightarrow \mathbb{R}$. Moreover, the flow of the Hamiltonian vector field $X_{H_{\mu}}$ on $\mathcal{S}_{\mu}$ is the $G_{\mu}$-quotient of the flow of $X_{H}$ on $J^{-1}(\mu)$. Let $c_{\mu}(t)$ be the integral curve of $X_{H_{\mu}}$ on $\mathcal{S}_{\mu}$. Then for $p_{0}\in J^{-1}(\mu)$ we will find the integral curve
$c(t)=\phi_{t}(p_{0})$   of $X_{H}$ such that $\pi_{\mu}(c(t))=c_{\mu}(t)$,
where $\pi_{\mu}:J^{-1}(\mu)\longrightarrow \mathcal{S}_{\mu}$ is a projection, see Figure 2. The following theorem is proved in \cite{MarsdenM}.
\begin{thm}
Suppose $\xymatrix{
G_{\mu} \ar@{^{(}->}[r] & J^{-1}(\mu)\ar[r]^{\pi} &\mathcal{S}_{\mu}}$ is a principal $G_{\mu} $-bundle with connection $A$. Moreover, let $c_{\mu}$ be the integral curve of the reduced dynamical system on $\mathcal{S}_{\mu}$. The integral curve of $\mathcal{S}$ through $p_{0}\in \pi^{-1}(c_{\mu}(0))$ is obtained as follows:
i) Horizontally lift $c_{\mu}$ to form the curve $d(t)\in  J^{-1}(\mu)$ through $p_{0}$; ii) Let  $\zeta (t)=A\cdot X_{H}(d(t))$, such that $\zeta(t)\in \mathbb{g}$; iii) Solve the equation $\dot{g}(t)=g(t)\cdot \zeta(t)$. Then $c(t)=g(t)\cdot d(t)$ is the integral curve of the system on $\mathcal{S}$ with initial condition $p_{0}$.
\end{thm}
\begin{example}
Consider the symplectic manifold $(\mathbb{C}^{n},\Omega)$, where
\begin{equation}
\Omega=\frac{i}{2}\sum_{i}
dz_{i}\wedge d\bar{z}_{i}=\sum_{i}
dx_{i}\wedge dy_{i}=\sum_{i}
r_{i}dr_{i}\wedge d\theta_{i}
\end{equation}
is a symplectic form on $\mathbb{C}^{n}$. The $S^{1}$-action on $(\mathbb{C}^{n},\Omega)$ is defined by $t \longmapsto \phi_{t}=\{\mathrm{multiplication by}  ~e^{it}\}$  is Hamiltonian with momentum map
 $J:\mathbb{C}^{n}\longrightarrow \mathbb{R}$ defined by $\psi\in \mathbb{C}^{n}\longmapsto -\frac{|\psi|^{2}}{2}+ C$, where $C$ is an arbitrary constant. If we chose the constant $C$ to be $C=\frac{1}{2}$, then $J^{-1}(0)=S^{2n-1}$ is the unit sphere in $\mathbb{C}^{n}$.
\end{example}

\section{Geometric formulation of  pure quantum states }\label{sec2}
In the geometric formulation of quantum mechanics  we consider a Hamiltonian dynamical system on a symplectic manifold, where the phase space is the projective Hilbert space constructed by principal fiber bundle and the evolution is governed by Schr\"{o}dinger's equation is equivalent to Hamilton's equations determined by symplectic structure. The K\"{a}hler structure of the quantum phase space includes a Riemannian metric that distinguishes the quantum from the classical mechanics.
In this section we will give an introduction to geometric quantum mechanics of pure states. The topics we will cover include quantum phase space, quantum dynamics, geometric uncertainty relation, quantum measurement, postulates of geometric quantum mechanics, and geometric phase.
\subsection{Quantum phase space }\label{QPS}
 In linear-algebraic approach to the quantum mechanics, a quantum  system is described on a  Hilbert space $\mathcal{H}$. We start by showing  that the Hilbert space is a K\"{a}hler space equipped with symplectic form and compatible Riemannian metric.
A hermitian inner product $
\langle \cdot\ket{\cdot}:\mathcal{H}\times \mathcal{H}\longrightarrow \mathbb{C}
$ on  a Hilbert space $\mathcal{H}$
is defined by
\begin{equation}
\langle \psi\ket{\phi}=\frac{1}{2\hbar}G(\psi, \phi)+\frac{i}{2\hbar}\Omega(\psi,\phi),
\end{equation}
for all  $\ket{\psi},\ket{\phi}\in \mathcal{H}$ where  the real part $G(\psi, \phi)=2 \hbar \mathrm{Re}\langle \psi\ket{\phi}$ is a Riemannian metric that satisfies the following relation
\begin{equation}
G(\psi, \phi)=G(\phi,\psi),
\end{equation}
and $\Omega(\psi, \phi)=2 \hbar \mathrm{Im}\langle \psi\ket{\phi}$ is a symplectic structure that satisfies
\begin{equation}
\Omega(\psi, \phi)=-\Omega(\phi,\psi).
\end{equation}
 If $\mathcal{H}$ is a complex Hilbert space, then an almost complex structure $J:\mathcal{H}\longrightarrow\mathcal{H}$ satisfies $J^2=-1$ and we have the following relations between $G$, $\Omega$ and $J$:
$
\Omega(\psi, \phi)=\Omega(J(\psi), J(\phi))
$ and $ G(\psi, \phi)=G(J(\psi), J(\phi))$. Moreover, since $ J(\psi)=i\psi$ it follows that
\begin{equation}
G(\psi, \phi)=\Omega(\psi,J(\phi))=-\Omega(J(\psi),\phi).
\end{equation}
These relations define a K\"{a}hler structure on $\mathcal{H}$. Thus the Hilbert space is a K\"{a}hler space.  Moreover, the Hilbert space is a symplectic manifold. Since $\mathcal{H}$ is isomorphic to its tangent space and symplectic form is non-degenerate, closed differential 2-form on $\mathcal{H}$.
\\
 For any state $\ket{\psi}\in \mathcal{H}$, the unit sphere in $\mathcal{H}$  is defined by
\begin{equation}
\mathcal{S}(\mathcal{H})=\{\ket{\psi}\in \mathcal{H}: \langle \psi\ket{\psi}=1\}\subset \mathcal{H}.
\end{equation}
Any two vectors $\ket{\psi}, \ket{\phi}\in \mathcal{S}(\mathcal{H})$ are equivalent if  they differ by a phase factor, that is $\ket{\psi}=e^{i\varphi} \ket{\phi}$, with $\varphi\in \mathbb{R}$.
Thus the proper phase space of pure quantum systems is
\begin{equation}
\mathcal{P}(\mathcal{H})=\mathcal{S}(\mathcal{H})/\sim,~~~~\textbf{quantum phase space}
\end{equation}
\begin{figure}[t]
\centering
\includegraphics[scale=.65]{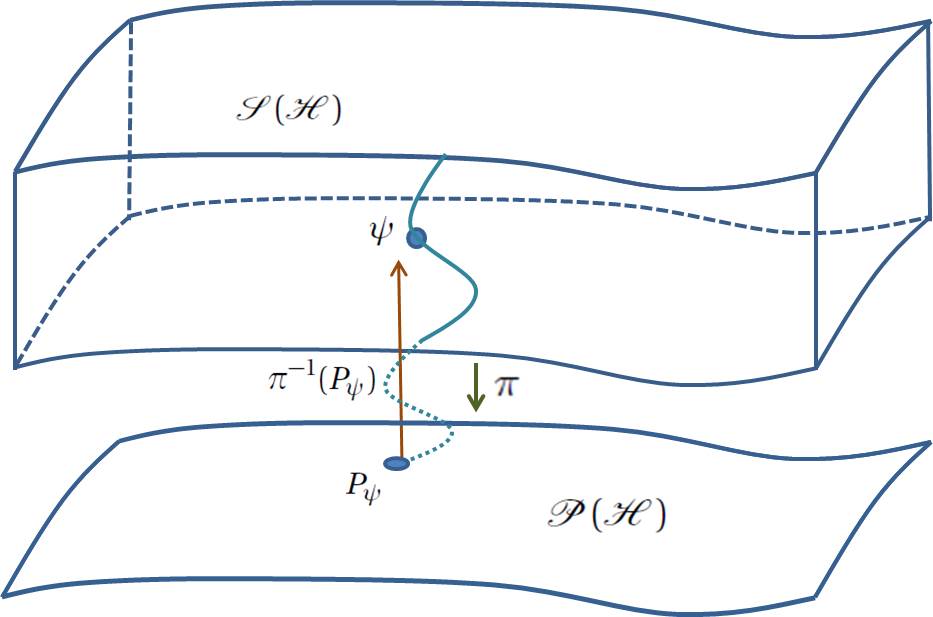}
%
%
\caption{Illustration of the principal fiber bundle.}
\label{fig11}      
\end{figure}
If $\ket{\psi}\in \mathcal{S}(\mathcal{H})$, then the corresponding equivalence class $[\ket{\psi}]$ can be identified with the one-dimensional projector
$
P_{\psi}=\ket{\psi}\bra{\psi}
$
which implies that $\mathcal{P}(\mathcal{H})$ is the space of one-dimensional projectors in $\mathcal{H}$.
This construction defines a principal $U(1)$-bundle
\begin{equation}
\xymatrix{
U(1) \ar@{^{(}->}[r] & \mathcal{S}(\mathcal{H})\ar[r]^{\pi} &\mathcal{P}(\mathcal{H})}
\end{equation}
over $\mathcal{P}(\mathcal{H})$, see Figure 3. Thus for any vector $\ket{\psi}\in \mathcal{H}$, the corresponding fibre
\begin{equation}
\pi^{-1}([\psi])=\left\{\frac{e^{i\varphi} \ket{\psi}}{\langle\psi\ket{\psi}}\in\mathcal{S}(\mathcal{H})\right\}
\end{equation}
could be identified with the Lie group $U(1)$.
In a finite dimensional quantum system the Hilbert space is given by $(n+1)$- dimensional Euclidean  space $\mathcal{H}=\mathbb{C}^{n+1}$ and we have the following principle fiber bundle
 \begin{equation}\label{Hopf}
\xymatrix{
U(1) \ar@{^{(}->}[r]&\mathbb{S}^{2n+1}\ar[r]^{\pi} &\mathbb{CP}^{n}},
\end{equation}
where the quantum phase space $\mathcal{P}(\mathcal{H})$ is a complex projective space
\begin{equation}
\mathbb{CP}^{n}=\mathbb{S}^{2n+1}/U(1),
\end{equation}
with $\mathbb{S}^{2n+1}$ being a unit sphere in $\mathbb{C}^{n+1}=\mathbb{R}^{2n+2}$.
 \begin{example}
 The simplest non-trivial case is the first Hopf bundle
 \begin{equation}
\xymatrix{
U(1)\cong\mathbb{S}^{1} \ar@{^{(}->}[r] &\mathbb{S}^{3}\ar[r]^{\pi} &\mathbb{CP}^{1}\cong \mathbb{S}^{2}}.
\end{equation}
Thus $\mathcal{P}(\mathbb{C}^{2})=\mathbb{CP}^{1}$ is the quantum phase space of  a quantum bit or a qubit state. To be able to define a qubit state explicitly,
let $\{\ket{i}\}^{1}_{i=0}$ be a set of orthonormal basis of a two-level quantum systems on $\mathcal{H}=\mathbb{C}^{2}$ and
\begin{eqnarray}
 \nonumber
 \ket{\psi} &=&\alpha_{0}\ket{0}+\alpha_{1}\ket{1}=
   e^{i\eta}\left(\cos\frac{\theta}{2}\ket{0}+e^{i\phi}\sin\frac{\theta}{2}\ket{1}\right)
\end{eqnarray}
 be a vector defined on $\mathcal{S}(\mathcal{H})=\mathbb{S}^{3}=\{(\alpha_{0},\alpha_{1})\in \mathbb{C}^{2}:|\alpha_{0}|^{2}+|\alpha_{1}|^{2}=1\}$. Then $\pi$ is defined by
\begin{equation} \label{coord}
\pi(\alpha_{0},\alpha_{1})=\mathrm{x}_{i}
=\langle\psi|\sigma_{i}\ket{\psi}=(\alpha^{*}_{0},\alpha^{*}_{1})
\sigma_{i}\left(%
\begin{array}{c}
  \alpha_{0}\\
  \alpha_{1}\\
\end{array}%
\right),
\end{equation}
 where $\sigma_{i}$ for $i=1,2,3$ are Pauli matrices,
\begin{equation}
\sigma=(\sigma_{1},\sigma_{2},\sigma_{3})=\left(\left(%
\begin{array}{cc}
  0 & 1 \\
  1 & 0 \\
\end{array}%
\right),\left(%
\begin{array}{cc}
  0 & -i \\
  i & 0 \\
\end{array}%
\right),\left(%
\begin{array}{cc}
  1 & 0 \\
  0 & -1 \\
\end{array}%
\right)\right).\end{equation}
 The sphere $\mathbb{S}^{2}$ is usually called the Bloch sphere
representation of the quantum bit $\ket{\psi}$.
 The simplest mixed quantum state is defined  by
\begin{equation}
\rho=\frac{1}{2}\left(I+\sum^{3}_{i=1}\mathrm{x}_{i}\sigma_{i}\right)=\left(%
\begin{array}{cc}
  1+\mathrm{x}_{3} & \mathrm{x}_{1}-i\mathrm{x}_{2} \\
  \mathrm{x}_{1}+i\mathrm{x}_{2} & 1-\mathrm{x}_{3} \\
\end{array}%
\right),\end{equation} where  $
  \mathrm{x}_{1}=\langle\psi|\sigma_{1}\ket{\psi}=2 \mathrm{Re}(\alpha_{0}^{*}\alpha_{1})$, $
  \mathrm{x}_{2}=\langle\psi|\sigma_{2}\ket{\psi}=2 \mathrm{Im}(\alpha^{*}_{0}\alpha_{1})$, and  $
  \mathrm{x}_{3}=\langle\psi|\sigma_{3}\ket{\psi}=|\alpha_{0}|^{2}-|\alpha_{1}|^{2}$  with the constraint $\det
\rho=1-\mathrm{x}^{2}_{1}-\mathrm{x}^{2}_{2}-\mathrm{x}^{2}_{3}\geq
0$. For a pure qubit state, we have equality, $\det \rho=0$. We will in details discuss the mixed quantum states in the next section.
\end{example}
Note that $\mathcal{P}(\mathcal{H})$ is equipped with a Hermitian structure induced by the one on $\mathcal{H}$ that makes $\mathcal{P}(\mathcal{H})$  a
 K\"{a}hler manifold. Next we will discuss K\"{a}hler structure on $\mathcal{P}(\mathcal{H})$.\\
The quantum phase space $\mathcal{P}(\mathcal{H})$ is a differentiable complex manifold and $\pi$ is  differentiable map. The tangent space $T_{[\psi]}\mathcal{P}(\mathcal{H})$ is isomorphic to quotient space $\mathcal{H}/\mathbb{C}\psi$, where $\mathbb{C}\psi$ is subspace of $\mathcal{H}$. Thus the projective map $\pi$ is a surjective submersion. Moreover, let $T_{\psi}\pi:\mathcal{H}\longrightarrow T_{[\psi]}\mathcal{P}(\mathcal{H})$. Then the kernel is defined by $\mathrm{Ker }T_{\psi}\pi =\mathbb{C}\psi$ and the restricted map $\mathrm{Ker }T_{\psi}\pi|_{(\mathbb{C}\psi)^{\perp}}:(\mathbb{C}\psi)^{\perp}\longrightarrow T_{[\psi]}\mathcal{P}(\mathcal{H})$ is a complex linear isomorphism  from $(\mathbb{C}\psi)^{\perp}=\{\ket{\psi}\in \mathcal{H}:\langle \psi\ket{\phi}=0\}$  to the tangent space of the quantum phase space $T_{[\psi]}\mathcal{P}(\mathcal{H})$ that also depends on chosen representative $\ket{\psi}\in [\psi]$.
\begin{prop}\label{FS}
Let $\ket{\psi}\in \mathcal{S}\mathcal{H}$ and $\ket{\phi_{1}}, \ket{\phi_{2}}\in (\mathbb{C}\psi)^{\perp}$, Then
\begin{equation}\label{HP}
\langle T_{\psi}\pi (\phi_{1})\ket{T_{\psi}\pi (\phi_{2})}=2\hbar\langle \phi_{1}\ket{ \phi_{2}}
\end{equation}
 gives a Hermitian inner product on $T_{[\psi]}\mathcal{P}(\mathcal{H})$, where the left hand side does not depend on the choice of $\ket{\psi}\in [\psi]$. Thus  (\ref{HP}) defines a  Hermitian metric on the quantum phase space  which is invariant under the action of transformation  $[U]$, for all unitary group $U$ on $\mathcal{H}$. Moreover,
\begin{equation}\label{HPo}
\omega( T_{\psi}\pi (\phi_{1}), T_{\psi}\pi (\phi_{2}))=2\hbar\mathrm{Im}\langle \phi_{1}\ket{ \phi_{2}}
\end{equation}
defines a symplectic form on quantum phase space.
Furthermore,
\begin{equation}\label{HPg}
g( T_{\psi}\pi (\phi_{1}), T_{\psi}\pi (\phi_{2}))=2\hbar\mathrm{Re}\langle \phi_{1}\ket{ \phi_{2}}
\end{equation}
defines a Riemannian metric on quantum phase space. The  symplectic form (\ref{HPo}) and the Riemannian metric (\ref{HPg}) are invariant under  all transformation   $[U]$.
\end{prop}
The Riemannian metric (\ref{HPg}) is usually called the Fubini-Study metric.  Let $\phi_{1},\phi_{2}\in \mathcal{H}$ and $\upsilon(\phi_{i})= T_{\psi}\pi (\phi_{i})$ for $i=1,2$. Moreover, assume that $\ket{\psi}\neq 0$. Then an explicit expression for the Hermitian metric  is given by
\begin{equation}\label{HP1}
\langle \upsilon(\phi_{1})\ket{\upsilon(\phi_{2})}=2\hbar\frac{1}{\|\ket{\psi}\|^{4}}\left(\|\ket{\psi}\|^{2}\langle \phi_{1}\ket{ \phi_{2}}-\langle \phi_{1}\ket{ \psi}\langle \psi\ket{ \phi_{2}}\right).
\end{equation}
Now the Riemannian metric  defined by equation (\ref{HPg}) is given by
\begin{equation}\label{HPo1}
g_{\psi}( \upsilon(\phi_{1}), \upsilon(\phi_{2}))=
2\hbar\frac{1}{\|\ket{\psi}\|^{4}}\mathrm{Re}\left(\|\ket{\psi}\|^{2}\langle \phi_{1}\ket{ \phi_{2}}-\langle \phi_{1}\ket{ \psi}\langle \psi\ket{ \phi_{2}}\right)
\end{equation}
and the symplectic form defined by equation (\ref{HPo}) is given by
\begin{equation}\label{HPg1}
\omega_{\psi}( \upsilon(\phi_{1}), \upsilon(\phi_{2}))=
2\hbar\frac{1}{\|\ket{\psi}\|^{4}}\mathrm{Im}\left(\|\ket{\psi}\|^{2}\langle \phi_{1}\ket{ \phi_{2}}-\langle \phi_{1}\ket{ \psi}\langle \psi\ket{ \phi_{2}}\right).
\end{equation}
For the proof of the proposition \ref{FS} see  reference \cite{Marsden_eta2000}.
Now, consider  the   principal $G_{\mu} $-bundle
\begin{equation}\xymatrix{
G_{\mu} \ar@{^{(}->}[r] & J^{-1}(\mu)\ar[r]^{\pi} &\mathcal{S}_{\mu}}
\end{equation}
that we have defined  in (\ref{mm}).  Let $\mathcal{S}=\mathcal{H}$
with the symplectic form $\Omega(\psi,\phi)=-\hbar \mathrm{Im} \langle \psi\ket{\phi}$ and the $\mathbb{S}^{1}$-action being the multiplication of a vector by $e^{i\theta}$. Then  $G=G_{\mu}=\mathbb{S}^{1}$ for any $\mu\in \mathbb{g}^{*}=\mathbb{R}$. If $\zeta\in \mathbb{R}$, then we have $\zeta_{\mathcal{H}}(\psi)=i\zeta\ket{\psi}$ for any $\ket{\psi}\in \mathcal{H}$. The momentum map $J:\mathcal{H}\longrightarrow \mathbb{R}$ is defined by
$J(\psi)=\frac{\|\ket{\psi}\|}{2}$. Moreover the symplectic form is given by $\Omega=- d\Theta$ with \begin{equation}\Theta(\psi)\cdot \ket{\phi}=\frac{\hbar}{2}\mathrm{Im}\langle \psi\ket{\phi}.
\end{equation}
Furthermore, the level set $J^{-1}\left(-\frac{1}{2}\right)$ is the sphere in $\mathcal{H}$ of radius one, that is $J^{-1}\left(-\frac{1}{2}\right)=\mathcal{S}(\mathcal{H})$. And finally the quantum phase space is $\mathcal{S}_{\mu}=\mathcal{P}(\mathcal{H})$.
\\
We can also apply the symplectic reduction theorem \ref{srt} to equip the quantum phase space with a symplectic form and a Riemannian metric as follows. Let $i:\mathcal{S}(\mathcal{H})\longrightarrow \mathcal{H}$. Then there is an unique symplectic form on $\mathcal{P}(\mathcal{H})$ such that $\pi^{*}\omega=i^{*}\Omega$. Thus we have
\begin{equation}\label{HPosrt}
\omega( \upsilon(\phi_{1}), \upsilon(\phi_{2}))=\Omega(\phi_{1}, \phi_{2}).
\end{equation}
Note that the expression given by
(\ref{HPo}) coincides with (\ref{HPosrt}) since for any vector $\phi\in  T_{\psi}\mathcal{S}(\mathcal{H})\cong(\mathbb{C}\psi)^{\perp}$ we have $\phi\perp\psi$. Thus for $\phi^{\perp}=\phi$ we have
\begin{equation}\label{HPosrt}
\omega( \upsilon(\phi_{1}), \upsilon(\phi_{2}))=2\hbar\mathrm{Im}\langle \phi_{1}\ket{ \phi_{2}}= \Omega(\phi_{1}, \phi_{2}).
\end{equation}
We also have $\pi^{*}g=i^{*}G$ which gives $g( \upsilon(\phi_{1}), \upsilon(\phi_{2}))=G(\phi_{1}, \phi_{2}).$
\\
We have in some details defined and characterized the quantum phase space of pure quantum states. In particular, we have used principal fiber bundle and momentum map to investigate the geometrical structures of quantum phase  space. In the next section we will discuss the quantum dynamics on quantum phase space based on the Hamiltonian dynamics.

\subsection{Quantum dynamics }\label{QPS}
The measurable quantities or observables  of the quantum system are represented by
hermitian/ self-adjoint linear operators acting on $\mathcal{H}$. One of the most important example of such an operator is called Hamiltonian operator $\hat{H}$ defined on $\mathcal{H}$.   Let $\ket{\psi}\in \mathcal{H}$. Then, the dynamics of quantum systems is described by the Schr\"{o}dinger's equation
\begin{equation}
i\hbar \frac{d \ket{\psi(t)}}{d t}=\hat{H}\ket{\psi(t)}.
\end{equation}
Let $\hat{A}$ be a hermitian/self-adjoint operator on $\mathcal{H}$. Then a real-valued expectation function $A:\mathcal{H}\longrightarrow \mathbb{R}$ is defined by
\begin{eqnarray}
 \nonumber
  A(\psi) &=& \frac{\langle \psi\ket{\hat{A}\psi}}{\langle \psi
  \ket{\psi}} \\
  &=& \frac{1}{2\hbar}G(\psi,\hat{A}\psi).
\end{eqnarray}
Now, we can associate to each $H$ a Hamiltonian vector field $X_{H}$, which is defined by $\imath_{X_{H}}\Omega=d H$. Moreover, we can identify $T_{\psi}\mathcal{H}$ with $\mathcal{H}$ since the Hilbert space is a linear space. Thus a vector field $X:\mathcal{H}\longrightarrow T_{\psi}\mathcal{H}$ can be identified with $X:\mathcal{H}\longrightarrow \mathcal{H}$ and a linear operator acts as a vector field on $\mathcal{H}$ as follows
\begin{equation}
X_{\hat{H}}(\psi)=-\frac{i}{\hbar}\hat{H}\psi.
\end{equation}
One also can show that an observable generates a 1-parameter group $\varphi_{t}:\mathcal{H}\longrightarrow \mathcal{H}$ defined by $\varphi_{t}(\psi)=e^{i \hat{H} t}\psi$ with $\psi\in T_{\psi}\mathcal{H}$.
\begin{thm}
The Schr\"{o}dinger vector field $X_{\hat{H}}$ is Hamiltonian and the Schr\"{o}dinger equation defines a classical Hamiltonian systems on $\mathcal{H}$:
\begin{equation}
X_{H}=\hbar X_{\hat{H}}.
\end{equation}
\end{thm}
\begin{proof}
We can identify the tangent space of a Hilbert space $T_{\psi}\mathcal{H}$ with the Hilbert space, since $\mathcal{H}$ is a linear space. Now, let $\psi, \xi\in \mathcal{H}$, then we have
\begin{eqnarray}
 \nonumber
 dH (\psi)(\xi) &=& \frac{d}{dt}H(\psi+ t \xi)\mid_{t=0}
 \\\nonumber&=&\frac{1}{2}\frac{d}{dt}
 \langle \psi+ t \xi\ket{\hat{H}(\psi+t \xi)}\mid_{t=0} \\\nonumber
  &=&  \frac{1}{2}\langle \xi\ket{\hat{H}\psi}+\frac{1}{2} \langle  \psi\ket{\hat{H} \xi}\\\nonumber
  &=& G(\xi,\hat{H}\psi)
  \\\nonumber
  &=& \hbar \Omega(\xi,\frac{i}{\hbar}\hat{H}\psi)\\\nonumber
  &=& \hbar \Omega( X_{\hat{H}},\xi)(\psi).
\end{eqnarray}
Thus $\hbar X_{\hat{H}}$ is a Hamiltonian vector field with respect to the symplectic form $\Omega$ on the Hilbert space.
\end{proof}
Next we will discuss the relation between Poisson bracket defined by the symplectic form
on the Hilbert space and the commutators
of quantum observables $\hat{A}$ and $\hat{B}$ with
corresponding expectation function $A$ and $B$ on $\mathcal{H}$ respectively. Let $X_{A}$ and $X_{B}$ be the
Hamiltonian vector fields corresponding to $A$ and $B$. Then we have
\begin{eqnarray}
 \nonumber
\{A,B\}_{\Omega}(\psi)&=& \Omega_{\psi}(X_{A},X_{B}) \\\nonumber
   &=& \frac{1}{2i}\left(\langle X_{A}(\psi)\ket{X_{B}(\psi)}
   -\langle X_{B}(\psi)\ket{X_{A}(\psi)}\right)\\\nonumber
   &=&\frac{1}{2i}\langle \psi\ket{(\hat{A}\hat{B}-\hat{B}\hat{A})\psi}
   \\\nonumber
   &=&\frac{1}{2i}\langle \psi\ket{[\hat{A},\hat{B}]\psi},
\end{eqnarray}
where  $X_{A}=\hbar X_{\hat{A}}=-i\hat{A}\psi$ and $X_{B}=\hbar X_{\hat{B}}=-i\hat{B}\psi$.
Let $P_{\psi}=\ket{\psi}\bra{\psi}\in \mathcal{P}(\mathcal{H})$ be a one-dimensional projector in $\mathcal{H}$ which is also called a density operator corresponding to the pure state $\ket{\psi}$.
 Then the evolution of $P_{\psi}$   under unitary operators is governed by von Neumann equation as follows
\begin{equation}
i\hbar \frac{d P_{\psi}}{dt}=[\hat{H}, P_{\psi}]
\end{equation}
with a solution that defines a curve in $\mathcal{P}(\mathcal{H})$.
Now, let $a:\mathcal{P}(\mathcal{H})\longrightarrow \mathbb{R}$ be a function defined by $a(P_{\psi})=A(\psi)$ or equivalently by $a(P_{\psi})=\mathrm{Tr}(\hat{A}P_{\psi})$.
We can also define a Hamiltonian vector field on quantum phase space $\mathcal{P}(\mathcal{H})$ by
$\imath_{H}\Omega=dH $, where $h(P_{\psi})=H(\psi)$. Since $\mathcal{P}(\mathcal{H})$ being a K\"{a}hler manifold, it is equipped with a symplectic from $\omega $ and the von Neumann equation can be written as follows
\begin{equation}
 \frac{d P_{\psi}}{dt}=\frac{1}{\hbar}\{h, P_{\psi}\}_{\omega},
\end{equation}
where $\{h, P_{\psi}\}_{\omega}$ is the Poisson bracket corresponding to the symplectic form $\omega$ on the quantum phase space. \\
Thus we have shown that in the geometric formulation of quantum mechanics the observables are real-valued functions and  Schr\"{o}dinger equation  is the symplectic flow of a Hamiltonian function on $\mathcal{P}(\mathcal{H})$. Moreover, the quantum phase space $\mathcal{P}(\mathcal{H})$ is a nonlinear manifold equipped with a K\"{a}hler structure and the flow generated by an observable consists of nonlinear symplectic transformation as in the classical mechanics.
\\
In the following section, we will discuss some applications of geometric quantum mechanics including geometric uncertainty relation, quantum measurement,  postulates of quantum mechanics, and geometric phase.
\subsection{Geometric uncertainty relation}
We have have shown that the quantum phase space $\mathcal{P}(\mathcal{H})$ of a pure state is  equipped with a symplectic and a Riemannian structure. Moreover, we have shown that the expectation values of
observables can be related to the Riemannian and symplectic structures.
This relation, enable us to derive a geometric version of uncertainty relation \cite{Robertson_1929}  for a pure state \cite{Chruscinski_etal2004}.
Let the uncertainty of an observable $\hat{C}$ corresponding to a normalized state $\ket{\psi}$ be defined by
\begin{eqnarray}
 (\Delta\hat{C})^{2}_{\psi} &=& \bra{\psi} \hat{C}^{2}\ket{\psi}-\bra{\psi} \hat{C}\ket{\psi}^{2}.
\end{eqnarray}
Then, the following theorem provides us with a geometric version of  uncertainty relation.
\begin{thm}
Let  $\hat{A}$ and $\hat{B}$  be two quantum observales on $\mathcal{H}$. Then we have
\begin{eqnarray}\label{inhs}
 (\Delta\hat{A})^{2}(\Delta\hat{B})^{2} &\geq&\Omega(X_{A},X_{B})
^{2}+\left(G(X_{A},X_{B})-AB
 \right)^{2},
\end{eqnarray}
where $(\Delta A)^{2}=(\Delta\hat{A})^{2}(\psi)$ is a function  on the Hilbert space $\mathcal{S}(\mathcal{H})$. Moreover, let $a:\hat{A}\longrightarrow \mathbb{R}$ and $b:\hat{B}\longrightarrow \mathbb{R}$ be two functions on $\mathcal{P}(\mathcal{H})$ of the observables $\hat{A}$ and $\hat{B}$ respectively, which are defined by
\begin{equation}
a \circ \pi=\bra{\psi}\hat{A} \ket{\psi}=A|_{\mathcal{S}(\mathcal{H})},~~
~b \circ \pi=\bra{\psi}\hat{B}\ket{\psi}=B|_{\mathcal{S}(\mathcal{H})}.
\end{equation}
Furthermore, let $g$ be the Riemannian metric and $\omega$ be the  symplectic form on $\mathcal{P}(\mathcal{H})$ such that the Poisson and Riemannian brackets can be defined  by
\begin{equation}
\{A,B\}_{\omega}=\omega(X_{A},X_{B}), ~~~(A,B)_{g}=g(X_{A},X_{B}),
\end{equation}
respectively. Then the uncertainty relation on $\mathcal{P}(\mathcal{H})$ is given by
\begin{eqnarray}
\label{inps}
 (\Delta a)^{2}(\Delta b)^{2} &\geq&\omega(X_{a},X_{b})^{2}
+g(X_{a},X_{b})^{2}\\\nonumber&=&
\{a,b\}^{2}_{\omega}+(a,b)^{2}_{g}
,
\end{eqnarray}
where $ (\Delta a)^{2}(P_{\psi})= (\Delta A)^{2}(\psi)$.
\end{thm}
\begin{proof}
Let
$
 \hat{C}_{\perp} =\hat{C}-\mathbf{1}C=\hat{C}-\mathbf{1}\bra{\psi}\hat{C}\ket{\psi},
$
where $C=\bra{\psi}\hat{C}\ket{\psi}$ is the expectation value of the observable $\hat{C}$ and $\mathbf{1}$ is an identity operator. Then it is easy to show that
$
 (\Delta\hat{C})^{2}_{\psi} =
 \bra{\Psi} \hat{C}^{2}_{\perp}\ket{\psi}
$.
Now for two quantum observables $\hat{A}$ and $\hat{B}$ we have
\begin{eqnarray}
 (\Delta\hat{A})^{2}_{\psi}(\Delta\hat{B})^{2}_{\psi} &=&
 \bra{\psi} \hat{A}^{2}_{\perp}\ket{\psi}\bra{\psi} \hat{B}^{2}_{\perp}\ket{\psi}.
\end{eqnarray}
And by Schwartz inequality we get
\begin{eqnarray}\label{anct1}
 \bra{\psi} \hat{A}^{2}_{\perp}\ket{\psi}\bra{\psi} \hat{B}^{2}_{\perp}\ket{\psi} &\geq&
 \bra{\psi} \hat{A}_{\perp}\hat{B}_{\perp}\ket{\psi}^{2}.
\end{eqnarray}
But we can also rewrite
\begin{eqnarray}\label{ancom1}
 \hat{A}_{\perp}\hat{B}_{\perp}=
 \frac{1}{2}[ \hat{A}_{\perp},\hat{B}_{\perp}]+\frac{1}{2}[ \hat{A}_{\perp},\hat{B}_{\perp}]_{+},
\end{eqnarray}
where $[ \hat{A}_{\perp},\hat{B}_{\perp}]_{+}= \hat{A}_{\perp}\hat{B}_{\perp} + \hat{B}_{\perp}\hat{A}_{\perp}$. Now by inserting equation (\ref{ancom1}) in equation (\ref{anct1}) we get
\begin{eqnarray}\label{in1}
 (\Delta\hat{A})^{2}_{\psi}(\Delta\hat{B})^{2}_{\psi} &\geq&\frac{1}{4}\left(
| \bra{\psi} [ \hat{A}_{\perp},\hat{B}_{\perp}]_{+}\ket{\psi}|^{2}+| \bra{\psi} [ \hat{A}_{\perp},\hat{B}_{\perp}]\ket{\Psi}|^{2}\right).
\end{eqnarray}
Note that $[ \hat{A}_{\perp},\hat{B}_{\perp}]=[ \hat{A},\hat{B}]$,
\begin{equation}
| \bra{\psi} [ \hat{A}_{\perp},\hat{B}_{\perp}]_{+}\ket{\psi}|^{2}=\bra{\psi} [ \hat{A}_{\perp},\hat{B}_{\perp}]\ket{\psi}^{2},
\end{equation}
and
$
| \bra{\psi} [ \hat{A},\hat{B}]\ket{\psi}|^{2}=-\bra{\psi} [ \hat{A},\hat{B}]\ket{\psi}^{2}
$,
which enable us to rewrite equation (\ref{anct1})  in the following form
\begin{eqnarray}
 (\Delta\hat{A})^{2}_{\psi}(\Delta\hat{B})^{2}_{\psi} &\geq&\frac{1}{4}\left(
 \bra{\psi} [ \hat{A}_{\perp},\hat{B}_{\perp}]_{+}\ket{\psi}^{2}- \bra{\psi} [ \hat{A},\hat{B}]\ket{\psi}^{2}\right).
\end{eqnarray}
Next, we will expand  $ [\hat{A}_{\perp},\hat{B}_{\perp}]_{+} $ as
\begin{eqnarray}
 [\hat{A}_{\perp},\hat{B}_{\perp}]_{+} &=&[\hat{A}-\mathbf{1}A,\hat{B}-\mathbf{1}B]_{+}
 \\\nonumber&=&
   \hat{A}\hat{B}+\hat{B}\hat{A}- \hat{A}B-A\hat{B}
 -\hat{B}A-B\hat{A}+2\mathbf{1}BA
 \\\nonumber&=&
  [\hat{A},\hat{B}]_{+}+2AB-2A\hat{B}-2B\hat{A}.
\end{eqnarray}
 Thus the  expectation value of $[\hat{A}_{\perp},\hat{B}_{\perp}]_{+}$ can be written  as
\begin{eqnarray}\nonumber
 \bra{\psi}[\hat{A}_{\perp},\hat{B}_{\perp}]_{+}\ket{\psi} &=&
  \bra{\psi}[\hat{A},\hat{B}]_{+}\ket{\psi}+2AB\langle \psi\ket{\psi}-2A\bra{\psi}\hat{B}
  \ket{\psi}-2B\bra{\psi}\hat{A}\ket{\psi}
   \\&=& \bra{\psi}[\hat{A},\hat{B}]_{+}\ket{\psi}-2AB.
\end{eqnarray}
Now, we want to rewrite uncertainty relation given by equation (\ref{anct1})  in terms of geometrical data we have, namely the symplectic structure
\begin{eqnarray}
\Omega_{\psi}(X_{A},X_{B})&=&\frac{1}{2i}\left(\bra{\psi}X_{A}X_{B}\ket{\psi}
-\bra{\psi}X_{B}X_{A}\ket{\psi}\right)
\\\nonumber&=&
\frac{1}{2i}\left(\bra{\Psi}\hat{A}\hat{B}-\hat{B}\hat{A}\ket{\Psi}\right)
\\\nonumber&=&
\frac{1}{2i}\bra{\psi}[\hat{A},\hat{B}]\ket{\psi},
\end{eqnarray}
where we have used $X_{A}=\hbar X_{\hat{A}}=-i\hat{A}\psi$ and similarly for $X_{B}$,  which implies that $\bra{\psi}[\hat{A},\hat{B}]\ket{\psi}=2i \Omega_{\psi}(X_{A},X_{B})$ and the Riemannian metric
\begin{eqnarray}
G_{\psi}(X_{A},X_{B})&=&\frac{1}{2}\left(\bra{\psi}X_{A}X_{B}\ket{\psi}
+\bra{\psi}X_{B}X_{A}\ket{\psi}\right)
\\\nonumber&=&
\frac{1}{2}\left(\bra{\psi}\hat{A}\hat{B}+\hat{B}\hat{A}\ket{\psi}\right)
\\\nonumber&=&
\frac{1}{2}\bra{\psi}[\hat{A},\hat{B}]_{+}\ket{\psi}
\end{eqnarray}
which also gives  $\bra{\psi}[\hat{A},\hat{B}]_{+}\ket{\psi}=2 G_{\psi}(X_{A},X_{B})$. Thus we have
\begin{eqnarray}
  \bra{\psi}[\hat{A}_{\perp},\hat{B}_{\perp}]_{+}\ket{\psi}&=&
  \bra{\psi}[\hat{A},\hat{B}]_{+}\ket{\psi}-2A B \\\nonumber&=&
  2 (G_{\psi}(X_{A},X_{B})-AB) .
\end{eqnarray}
The inequality (\ref{in1}) can now be written in terms of the Riemannian metric and the symplectic form on the Hilbert space as follows
\begin{eqnarray}
 \nonumber
 (\Delta\hat{A})^{2}(\Delta\hat{B})^{2} &\geq&\Omega(X_{A},X_{B})
^{2}+\left(G(X_{A},X_{B})-AB
 \right)^{2}.
\end{eqnarray}
This proof the first part of theorem.
Next we want to prove (\ref{inps}) which gives  an uncertainty relation on the quantum phase space $\mathcal{P}(\mathcal{H})$.  If $\xi,\eta\in T_{p}\mathcal{P}(\mathcal{H})$, then
\begin{equation}
g_{p}(\xi,\eta)=G_{\psi}(X^{\perp},Y^{\perp}),~~~~\omega_{p}(\xi,\eta)
=\Omega_{\psi}(X^{\perp},Y^{\perp}),
\end{equation}
and $\psi\in \mathcal{H}$ is projected to $p\in T_{p}\mathcal{P}(\mathcal{H})$.
If $X$ and $Y$ are arbitrary vectors in $\mathcal{H}$, then
\begin{equation}
X^{\perp}=X-\frac{\langle\psi\ket{X}}{\langle\psi\ket{\psi}}\psi,~~~~
~Y^{\perp}=Y-\frac{\langle\psi\ket{Y}}{\langle\psi\ket{\psi}}\psi.
\end{equation}
Thus we have $g_{p}(X_{a},X_{b})=G_{\psi}(X^{\perp}_{A},Y^{\perp}_{B})$ and $\omega_{p}(X_{a},X_{b})
=\Omega_{\psi}(X^{\perp}_{A},Y^{\perp}_{B})$
which gives
\begin{equation}
g_{p}(X_{a},X_{b})=G_{\psi}(X^{\perp}_{A},Y^{\perp}_{B})=G_{\psi}(X_{A},Y_{B})-AB(\psi)
\end{equation}
and
\begin{equation}\omega_{p}(X_{a},X_{b})=\Omega_{\psi}(X^{\perp}_{A},Y^{\perp}_{B})
=\Omega_{\psi}(X_{A},Y_{B}).
\end{equation}
Now the uncertainty relation on $\mathcal{P}(\mathcal{H})$ can be written as
\begin{eqnarray}
 (\Delta a)^{2}(\Delta b)^{2} &\geq&\omega(X_{a},X_{b})^{2}
+g(X_{a},X_{b})^{2}=
\{a,b\}^{2}_{\omega}+(a,b)^{2}_{g}
.
\end{eqnarray}
This end the proof of our theorem on geometric uncertainty relation for pure quantum states.
\end{proof}
 Note that in a special case we have
\begin{equation} (\Delta a)^{2}=g(X_{a},X_{a})
\end{equation}
which gives rise to a geometrical interpretation of quantum uncertainty relation. Let $X_{H}$ be a Hamiltonian vector field. Then the  uncertainty of the energy of a quantum system
\begin{equation} (\Delta h)^{2}=g(X_{h},X_{h})
\end{equation}
is equal to the length of $X_{h}$. This establishes a direct relation between measurable quantity of a physical quantum system and geometry  of underling phase space. In particular  the energy uncertainty   measures the speed  at which the quantum system  travels through quantum phase space. For applications of this result to quantum speed limit see references \cite{Bekenstein1981,Hegerfeldt2013,QSL}.
This end our detail discussion of geometric uncertainly relation for pure quantum states based on symplectic and Riemannian structures of the Hilbert space and the quantum phase space $\mathcal{P}(\mathcal{H})$. In the following sections, we will  also derive a geometric uncertainly relation and quantum energy dispersion relation  for mixed quantum states.

\subsection{Quantum measurement}
Since the quantum phase space of a pure state is a K\"{a}hler manifold  which is equipped with the Fubini-Study metric, it enables one to  measure distance between two quantum states. In this section we will give a short introduction to quantum  distance measure and its application to quantum measurement process.
\\ Consider the principal fiber bundle $\xymatrix{
U(1) \ar@{^{(}->}[r] & \mathcal{S}(\mathcal{H})\ar[r]^{\pi} &\mathcal{P}(\mathcal{H})}. $ Let $\ket{\psi}\in \mathcal{S}(\mathcal{H})$ and  $C(t)$ be a curve in $\mathcal{P}(\mathcal{H})$ defined by $c(t)\in \mathcal{P}(\mathcal{H})$ for all $t\in[0,1]$. Moreover, let $\tilde{C}$ be its lift defined by $\psi(t)\in \mathcal{P}(\mathcal{H})$  for all $t\in[0,1]$. Then the length of the lift is determined by
\begin{equation}
\mathrm{Length}(\tilde{C})=\int^{t=1}_{t=0}\sqrt{\langle\dot{\psi}\ket{\dot{\psi}}}dt.
\end{equation}
We could wonder  which lift of the curve has minimal length or  is a minimal lift. This question is answered in the following theorem.
 \begin{thm}
A lift is minimal if and only if it is a horizontal lift. The length of the curve $C$ measured by Fubini-Study metric is equal to the length of the horizontal lift of the curve.
\end{thm}
\begin{proof}
The proof follows from the following inequality
\begin{eqnarray}
\nonumber
\langle \dot{\psi}\ket{\dot{\psi}}\geq \langle \dot{\psi}\ket{\dot{\psi}}-|\langle \dot{\psi}\ket{\psi}|^{2}
\end{eqnarray}
which implies that $\langle \dot{\psi}\ket{\dot{\psi}}\geq(\pi^{*}g)_{\psi}
(\dot{\psi},\dot{\psi}),$
where $g_{\psi}(\dot{\psi},\dot{\psi})$ is the Fubini-Study metric.
So the lift is minimal if and only if $\langle \psi\ket{\dot{\psi}}=0$. Moreover, if $\ket{\psi(t)}$ is a horizontal lift then we have
$\langle \dot{\psi}\ket{\dot{\psi}}=(\pi^{*}g)_{\psi}(\dot{\psi},\dot{\psi})$.
\end{proof}
Let $p_{1},p_{2}\in\mathcal{P}(\mathcal{H})$ and $\gamma$ be the shortest geodesic joining them. Moreover, let $\kappa(p_{1},p_{2})$ be equal to the length of $\gamma$. Then we have the following theorem \cite{Chruscinski_etal2004}:
\begin{thm}
 Let $\ket{\psi_{i}}$ be arbitrary elements from the fibers
  $\pi^{-1}(p_{i})$ for all $i=1,2$. Then the length of $\gamma$ is given by
\begin{eqnarray}
\nonumber
|\langle \psi_{1}\ket{\psi_{2}}|=\cos[\kappa(p_{1},p_{2})].
\end{eqnarray}
\end{thm}
Note that if we choose $\ket{\psi_{i}}$ such that $|\langle \psi_{1}\ket{\psi_{2}}|$ is real and positive, then we will have
$\langle \psi_{1}\ket{\psi_{2}}=\cos[\kappa(p_{1},p_{2})].$
Moreover,  every geodesic on the quantum phase space $\mathcal{P}(\mathcal{H})$ is closed and its length  is equals $\pi$  which  is also half of the length of a closed geodesic in $\mathcal{S}(\mathcal{H})$. Furthermore, if $\langle \psi_{1}\ket{\psi_{2}}=0$ then the length of the shortest geodesic in $\mathcal{P}(\mathcal{H})$ will be $\pi/2$. In this case there are finitely many planes spanned by $ \ket{\psi_{1}}$ and $\ket{\psi_{2}}$. Thus there are finitely many geodesics connecting the points $p_{1}$ and $p_{2}$ in $\mathcal{P}(\mathcal{H})$. These points are usually called the conjugated points. This mean that if two vectors are orthogonal in $\mathcal{S}(\mathcal{H})$,  then they give rise to two conjugated points in $\mathcal{P}(\mathcal{H})$.
\\
Let $p_{1},p_{2}\in\mathcal{P}(\mathcal{H})$. Then the Fubini-Study length between these two points equals the length of the shortest geodesic  $\kappa(p_{1},p_{2})$.
However, if  we want to measure the distance on the Hilbert space $\mathcal{H}$, then the distance measure is called the Fubini-Study distance and is defined by
\begin{equation}
D_{FS}(\ket{\psi_{1}},\ket{\psi_{2}})=\inf_{\varphi}\|\ket{\psi_{1}}-e^{i\varphi}\ket{\psi_{2}}\|
=\sqrt{2(1-|\langle \psi_{1}\ket{\psi_{2}}|)}.
\end{equation}
There are another ways to compute the distance between pure quantum states. The most well-known one is called the trace distance and it is defined by
\begin{equation}
D(\ket{\psi_{1}},\ket{\psi_{2}})=\mathrm{Tr}|\ket{\psi_{1}}\bra{\psi_{1}}-\ket{\psi_{2}}\bra{\psi_{2}}|
=2\sqrt{1-|\langle \psi_{1}\ket{\psi_{2}}|^{2}},
\end{equation}
where $|M|=\sqrt{M^{\dagger}M}$. The second measure of distance that we will consider is called the Hilbert-Schmidt norm and it is defined by
\begin{equation}
D_{HS}(\ket{\psi_{1}},\ket{\psi_{2}})=\sqrt{\mathrm{Tr}(\ket{\psi_{1}}\bra{\psi_{1}}-\ket{\psi_{2}}\bra{\psi_{2}})^{2}}
=\sqrt{2(1-|\langle \psi_{1}\ket{\psi_{2}}|^{2})}.
\end{equation}
We can see that these measures of quantum distance are related to the geodesic length of a curve on $\mathcal{P}(\mathcal{H})$, since $|\langle \psi_{1}\ket{\psi_{2}}|= \arccos[\kappa(p_{1},p_{2})]$.

Next we will discuss the quantum measurement process and relates it to geodesic distance on quantum phase space.
Let $p_{0},p_{1}\in \mathcal{P}(\mathcal{H})$ with corresponding fibers $\ket{\psi_{0}}\in \pi^{-1}(p_{0})$ and $\ket{\psi}\in\pi^{-1}(p)$ defined on $ \mathcal{S}(\mathcal{H})$. Then one can consider a function $\delta_{0}:\mathcal{P}(\mathcal{H})\longrightarrow \mathbb{R}_{+}$  defined by
 \begin{equation}\delta_{0}(p)=|\langle \psi_{0}\ket{\psi}|^{2}
\end{equation}
which is called the quantum probability distribution on $\mathcal{P}(\mathcal{H})$. Moreover, if the projectors $P_{0}=\ket{\psi_{0}}\bra{\psi_{0}}$ and $P=\ket{\psi}\bra{\psi}$ are corresponding to the points $p_{0}$ and $p$, then we also have
\begin{equation}
\delta_{0}(p)=\mathrm{Tr}(P_{0}P).
\end{equation}
Now, if $\kappa(p_{0},p)$ is the length of the minimal geodesic distance separating $p_{0},p \in\mathcal{P}(\mathcal{H})$, then the quantum mechanical probability distribution on quantum phase space satisfies \begin{equation}
\delta_{0}(p)=\cos^{2}[\frac{\kappa(p_{0},p)}{\sqrt{2\hbar}}].
\end{equation}
Consider again the principal fiber bundle $\xymatrix{
U(1) \ar@{^{(}->}[r] & \mathcal{S}(\mathcal{H})\ar[r]^{\pi} &\mathcal{P}(\mathcal{H})}$.
 Let $\hat{F}$ be a self-adjoint/hermitian operator on the Hilbert space which has discrete  non-degenerated spectrum, that is $F\ket{\psi_{k}}=f_{k}\ket{\psi_{k}}$, where $\ket{\psi_{k}}\in \mathcal{S}\mathcal{H}$ and $p_{k}\in  \mathcal{P}\mathcal{H}$  be corresponding eigenstates with $p_{k}=\pi (\ket{\psi_{k}})$.
If we define $P_{k}=\ket{\psi_{k}}\bra {\psi_{k}}$ as the one-dimensional projection, then the  observable $\hat{F}$ has  the following spectral decomposition $\hat{F}=\sum_{k}f_{k}P_{k}$. Thus if we measure the quantum system, then any state $P_{0}$ in quantum phase space will collapse to one of $P_{k}$ with following probability $\delta_{0}(P_{k})=\mathrm{Tr}(P_{0}P_{k})$. We  can argue that in the process of a quantum measurement  of an observable $\hat{F}$, the probability of obtaining an eigenvalue  $f_{k}$ is a monotonically decreasing function of $P_{0}$ and $ P_{k}$. Now, let $\hat{F}$ be a self-adjoint/hermitian operator on the Hilbert space and $f:\mathcal{P}(\mathcal{H})\longrightarrow \mathbb{R}$ be a function defined by  $f(P)=\mathrm{Tr}(P\hat{F})$ with $f(P_{k})=f_{k}$. Then $X_{\hat{F}}(\psi)=-\frac{i}{\hbar}\hat{F} \ket{\psi}$ is a Hamiltonian vector field on the Hilbert space with $X_{\hat{F}}(\psi_{k})=-\frac{i}{\hbar}\hat{F} \ket{\psi_{k}}$. Thus $X_{f}$ will vanishes at all eigenstates  $P_{k}$.  The conclusion is that the eigenstates $p_{k}$ are the critical points of the observable function $f$.
Thus the observable function $f$ can be determined based on geometric structure of quantum phase space.
\begin{prop}\label{qobs}
 $f$ is a quantum observable if and only if corresponding Hamiltonian vector $X_{f}$ is Killing vector field of the K\"{a}hler metric $g$, that is $\mathcal{L}_{X_{f}}g=0$.
\end{prop}
For more information on the case when the observables have continues spectra  see   \cite{Ashtekar_etal1998}.
\subsection{Postulates of geometric quantum mechanics}

As we have shown, the true  space of the quantum states is a K\"{a}hler manifold  $\mathcal{P}(\mathcal{H})$, the states are represented
by points of $\mathcal{P}(\mathcal{H})$ equipped with a symplectic form and  a Riemannian metric. Moreover, the observables are represented by certain real-valued
functions on  $\mathcal{P}(\mathcal{H})$ and the Schr\"{o}dinger evolution is captured by
the symplectic flow generated by a Hamiltonian function  on  $\mathcal{P}(\mathcal{H})$.
There is thus a remarkable
similarity with the standard symplectic formulation of classical mechanics.
 Thus we can give a set of postulates of quantum mechanics based on the structures of quantum phase space $\mathcal{P}(\mathcal{H})$.
We will assume that observables are hermitian operators. Moreover, if the Hilbert space is finite-dimensional, then the spectrum of an operator $\hat{A}$ consists of the eigenvalues of $\hat{A}$ and it can be written as
\begin{equation}
\hat{A}=\sum_{i}\alpha_{i}P_{\hat{A},\alpha_{i}},
\end{equation}
where $P_{\hat{A},\alpha_{i}}$ is the projection operator corresponding to the eigenvalue $\alpha_{i}$ of $\hat{A}$. Here is the summary of the geometric postulates of quantum mechanics.
\begin{itemize}
  \item \textbf{Physical state:} In quantum mechanics, the physical states of the systems are in one-one correspondence with the points of quantum phase space $\mathcal{P}(\mathcal{H})$.
  \item \textbf{Quantum evolution:} The quantum unitary evolution is given by the flow on the quantum phase space $\mathcal{P}(\mathcal{H})$. The flow preserves the K\"{a}hler structure and the generator of the flow is densely defined vector  field on $\mathcal{P}(\mathcal{H})$.
  \item \textbf{Observables:} Quantum observables are presented by smooth and real-valued function $f:\mathcal{P}(\mathcal{H})\longrightarrow\mathbb{R}$. The  flow of Hamiltonian vector fields $X_{f}$ corresponding to the function $f$ preserves the K\"{a}hler structure.
      \item \textbf{Probabilistic interpretation:} Let $f$ be an observable and $\Sigma \subset \mathbb{R}$ be a closed subset of spectrum of $f$ defined by $sp(f)=\{\lambda\in\mathbb{R}| n_{\lambda}:\mathcal{P}(\mathcal{H})\longrightarrow \mathbb{R}\cup\infty$, defined by $p\longmapsto[(\Delta f)^{2}(p)+(f(p)-\lambda)^{2}]^{-1}$ is unbounded $\}$. Moreover, suppose that the system is in the states corresponding to $p\in\mathcal{P}(\mathcal{H})$. Then the measurement of $f$ will give an element of $\Sigma$ with the probability
          \begin{equation}
          \delta_{p}(\Sigma)=\cos^{2}
          \left(\frac{1}{\sqrt{2\hbar}}\kappa(p,P_{f,\Sigma}(p))\right),
          \end{equation}
          where  $\kappa(p,P_{f,\Sigma}(p))$ is minimal geodesics distance between $p$ and $P_{f,\Sigma}(p)$. Moreover, $P_{f,\Sigma}(p)$ is  the point closest to $p$ in \begin{equation} E_{f,\Sigma}=\{q\in\mathcal{P}(\mathcal{H}): \{f,\{f,\{f,\ldots\}_{+}\}_{+} \}_{+}|_{q}\in \Sigma^{n}~\forall n\},\end{equation} where $\Sigma^{n}$ denotes the image of $\Sigma$ under the mapping $\Sigma \longmapsto\Sigma^{n}$ and  e.g., $\{f,f\}_{+}$ is the expectation value of observable $\hat{F}^{2}$ corresponding to $f$.
      \item\textbf{Reduction of quantum states:} The discrete spectrum of  an observable $f$ provides the set of possible outcomes of the measurement of $f$. The state of the system after measurement of an observable $f$ with corresponding eigenvalue $\lambda$  is given by the associated projection $P_{f,\lambda}(p)$ of the initial state $p$.
\end{itemize}
In the above geometric postulates of quantum mechanics we didn't refer to the linear structure of the Hilbert space which provides essential mathematical tools in linear and algebraic formulation of quantum mechanics. We could wonder if the geometric formulation of quantum mechanics will some day provides us with a generalization of quantum mechanics. Such a generalization of quantum mechanics may includes: i) the quantum phase space, ii) the algebra of the quantum observables, and iii) the quantum dynamics.
More information on the geometric postulates of quantum mechanics can be found in  \cite{Ashtekar_etal1998}.

\subsection{Geometric phase a fiber bundle approach} \label{GP} Geometric phases have many applications in different fields of quantum physics such as quantum computation and condensed matter physics \cite{Berry1984,Shapere_etal1989,Simon1983}. In this section we give a short introduction to geometric phase based on a fiber bundle approach.
Consider the principal fiber bundle $\xymatrix{
U(1) \ar@{^{(}->}[r] & \mathcal{S}(\mathcal{H})\ar[r]^{\pi} &\mathcal{P}(\mathcal{H})}$. Let $\ket{\psi}\in \mathcal{S}(\mathcal{H})$. Then the tangent space can be decompose as
$T_{\psi}\mathcal{S}(\mathcal{H})=V_{\psi}\mathcal{S}(\mathcal{H})\oplus H_{\psi}\mathcal{S}(\mathcal{H})$. A fiber $\pi^{-1}(\psi)$ can be written as $e^{i\varphi}\ket{\psi}$. Thus we can define the vertical tangent bundle as
\begin{equation}V_{\psi}\mathcal{S}(\mathcal{H})=\{i\varphi \ket{\psi}:\varphi\in \mathbb{R}\}
\end{equation}
which can be identified with the Lie algebra $\mathfrak{u}(1)\cong i \mathbb{R}$. Moreover, we define the horizontal  tangent bundle as
\begin{equation}
H_{\psi}\mathcal{S}(\mathcal{H})=\{X\in \mathcal{H}:\langle \psi\ket{X}=0\}.
\end{equation}
 A curve $t\longmapsto \ket{\psi(t)}\in \mathcal{S}(\mathcal{H})$ is horizontal if $\langle \psi(t)\ket{\dot{\psi}(t)}=0$ for all $t$. Now, we define a connection one-form on $\mathcal{S}(\mathcal{H})$ by $\mathcal{A}_{\psi}=i\mathrm{Im}\langle \psi\ket{X}\in\mathfrak{u}(1)$. Let
 \begin{equation}
 s:\mathcal{P}(\mathcal{H})\longrightarrow\mathcal{S}(\mathcal{H})
\end{equation}
be a local section. Then a local connection form on $\mathcal{P}(\mathcal{H})$ is defined by pull back of $s$, that is $A=i s^{*} \mathcal{A}$. An explicit form for  $A$ is given by $A=i\langle \psi\ket{d\psi}$. Now, for a closed curve $C$ in $\mathcal{P}(\mathcal{H})$ the holonomy is defined by
\begin{equation}
\mathrm{Hol}(C,\psi)=\exp\left(i\oint_{C}A\right)
\end{equation}
which coincides with Aharonov-Anandan \cite{Anandan_etal1990} phase factor.
 We will discuss a generalization of this approach to mixed quantum states in the following section.

\section{Geometric formulation of  mixed quantum states }\label{sec3}
\subsection{Introduction}
Pure quantum states are small subclasses of all quantum states. Mixed quantum states represented by density operators $\rho$ are the most general quantum states in quantum mechanics.  These generalized quantum states are hermitian trace class operators acting on  the Hilbert space with the following properties:
i) $\rho\geq0$ and ii) $\mathrm{Tr} (\rho)=1$. Let
\begin{equation}
\lambda_{1}\geq \lambda_{2}\geq\cdots\geq \lambda_{k}
\end{equation}
be the eigenvalues of $\rho$ with multiplicity $(m_{1}, m_{2},\ldots,m_{k})$.
Then the spectral decomposition  $\varpi$ of $\rho$ can be written as
\begin{equation}
\varpi=\left(
         \begin{array}{cccc}
           \lambda_{1}\mathbf{1}_{m_{1}} & 0 & \cdots &0\\
           0& \lambda_{2}\mathbf{1}_{m_{2}} & 0&0 \\
          \vdots& 0&\ddots & 0 \\
           0&0 & \cdots & \lambda_{k}\mathbf{1}_{m_{k}} \\
         \end{array}
       \right),
\end{equation}
where $\mathbf{1}_{m_{i}}$ are $m_{i}\times m_{i}$ the identity matrices for all $1\leq i\leq k$. Now let $U\in U(n)$ be a unitary matrix, where $n=m_{1}+m_{2}+\cdots +m_{k}$. Then all density matrices unitary equivalent to $\varpi$, that is  $\rho=U\varpi U^{\dagger}$ lie on the co-adjoint orbit passing through $\varpi$ defined by
\begin{equation}
\mathcal{O}_{\varpi}=\{U\varpi U^{\dagger}:U\in U(n)\}\cong U(n)/G_{\varpi},
\end{equation}
where $G_{\varpi}=\{U\in U(n):U\varpi U^{\dagger}=\varpi\}$ is the isotropy subgroup of $\varpi$. Equivalently we can define $\mathcal{O}_{\varpi}$ as
\begin{equation}
\mathcal{O}_{\varpi}\cong U(n)/U(m_{1})\times U(m_{2})\times\cdots\times U(m_{k}),
\end{equation}
where we have identified $G_{\varpi}$ as $G_{\varpi}\cong U(m_{1})\times U(m_{2})\times\cdots\times U(m_{k})$.
As we have discussed in example \ref{hs},  a homogenous manifold can  be construct by a principal fiber bundle. Thus the space of  mixed quantum states can be constructed by the following principal fiber bundle
\begin{equation}
\xymatrix{
G_{\varpi} \ar@{^{(}->}[r] & U(n)\ar[r]^{\pi} &\mathcal{O}_{\varpi}}
\end{equation}
or equivalently by
\begin{equation}
\xymatrix{
U(m_{1})\times\cdots\times U(m_{k})\ar@{^{(}->}[r] & U(n)\ar[r]^{\pi~~~~~~~~~~~~} &U(n)/U(m_{1})\times\cdots\times U(m_{k})}.
\end{equation}
This definition of $\mathcal{O}_{\varpi}$ indicates that the quantum phase space of mixed quantum states is a complex flag manifold, which  is usually denoted by $\mathbb{CF}_{m_{1}\cdots m_{k}}$.
 \begin{example}
 The first non-trivial example of such a space is called a complex Grassmann manifold $\mathbb{CF}_{m_{1}m_{2}}=\mathbb{CG}_{m,n-m}$ which is defined by e.g., taking $m_{1}=m$ and $m_{2}=n-m$, that is
\begin{equation}
\mathbb{CG}_{m,n-m}\cong U(n)/U(m)\times U(n-m)
\end{equation}
or in terms of a principal fiber bundle
\begin{equation}
\xymatrix{
U(m)\times U(n-m)\ar@{^{(}->}[r] & U(n)\ar[r]^{\pi~~~~~~~~~~~~~~~~~~~~~~~~} &~\mathbb{CG}_{m,n-m}\cong U(n)/U(m)\times U(n-m)}.
\end{equation}
\end{example}
Now, it is  possible to introduce a geometric framework for mixed quantum states based on a K\"ahler structure. The geometric framework includes a symplectic form, an almost complex structure, and a Riemannian metric that characterize the space of mixed quantum states \cite{Hosh}. The framework is computationally effective  and it provides us with a better understanding of general quantum mechanical systems. However, in the next section we will review another geometric framework for mixed quantum states based on principal fiber bundle with some important applications in geometric quantum mechanics.

\subsection{A geometric framework for mixed quantum states based on a principal fiber bundle}

Recently, we have introduced a geometric formulation of quantum mechanics for density operators based on principal fiber bundle and purification procedure which has led to many interesting results such as a geometric phases, an uncertainty relations, quantum speed limits, a distance measure, and an optimal Hamiltonian \cite{GP,MB,DD,GQE,GUR,QSL}.  Let $\mathcal{H}$ be a  Hilbert space. In this work we will consider that $\mathcal{H}$  is finite dimensional but the theory can carefully be extended to infinite dimensional cases.
Then  $\mathcal{D}(\mathcal{H})$ will denote the space of density operators on $\mathcal{H}$. Moreover, we let $\mathcal{D}_{k}(\mathcal{H})$ be the space of density operators on $\mathcal{H}$ which has finite rank, namely less than or equal to $k$.
\subsubsection{Purification}
In the pervious section we have covered  geometric formulation of pure states which are density operators with 1-dimensional support. In this section we will consider the density operator with $k$-dimensional support where $k$ is a positive integer. However, every density operator can be regarded as a reduced pure state. Let $\mathcal{K}$  be a $k$-dimensional Hilbert space. Moreover, let $\mathcal{L}(\mathcal{K},\mathcal{H})\cong \mathcal{H}\otimes\mathcal{K}^{*}$ be the space of linear mapping from $\mathcal{K}$ to $\mathcal{H}$ and $\mathcal{S}(\mathcal{K},\mathcal{H})$ be the space of unit sphere in $\mathcal{L}(\mathcal{K},\mathcal{H})$. Now a purification of density operator can be defined by the following surjective map
\begin{equation}
\mathcal{S}(\mathcal{K},\mathcal{H})\longrightarrow\mathcal{D}_{k}(\mathcal{H})
\end{equation}
defined by $\psi\longmapsto \psi \psi^{\dagger}$. The idea of the purification is based on the fact that a quantum system defined on the Hilbert space $\mathcal{H}$  can be consider as a subsystem of a larger quantum system $\mathcal{H}\otimes\mathcal{K}$. If $\rho$ is a density operator on $\mathcal{H}$, then it can be defined by the following partial trace $\rho=\mathrm{Tr}_{\mathcal{K}}(\varsigma)$, where $\varsigma$ is  a density operator on $\mathcal{H}\otimes\mathcal{K}$. In our case if we consider $\mathcal{P}(\mathcal{H}\otimes\mathcal{K}^{*})$ as projective space over $\mathcal{H}\otimes\mathcal{K}^{*}$, then we have
\begin{equation}
\xymatrix{
\mathcal{S}(\mathcal{H}\otimes\mathcal{K}^{*}) \ar@{^{(}->}[r]^{\mathcal{C}} & \mathcal{P}(\mathcal{H}\otimes\mathcal{K}^{*})\ar[r]^{\mathrm{Tr}_{\mathcal{K}^*}} &\mathcal{D}_{k}(\mathcal{H})},
\end{equation}
where $\mathcal{C}:\ket{\psi}\mapsto\ket{\psi}\bra{\psi}$. Now, let $\mathcal{U}(\mathcal{H})$ be the unitary group of $\mathcal{H}$ acting on  $\mathcal{D}(\mathcal{H})$. Then the evolution of density operators which is governed by a von Neumann equation will stay in a single orbit for the left conjugation-action of $\mathcal{U}(\mathcal{H})$. In this setting the orbits are in one-to-one correspondence with the spectra of density operators on $\mathcal{H}$.
Let $\lambda_{i}$ be the density operator's  eigenvalues with multiplicities $m_{i}$ listed in descending order. Then we define the spectrum of the density operator by
\begin{equation}
\sigma=(\lambda_{1},\lambda_{2},\ldots,\lambda_{k};m_{1},m_{2},\ldots,m_{k}).
\end{equation}
Given the spectrum $\sigma$, we denote the corresponding $\mathcal{U}(\mathcal{H})$ orbit in $\mathcal{D}_{k}(\mathcal{H})$ by $\mathcal{D}(\sigma)$. Now, let $\{\ket{j}\}^{k}_{j=1}$ be a set of orthonormal basis in $\mathcal{K}$. Then we define
\begin{equation}
P(\sigma)=\sum^{k}_{j=1}\lambda_{j}\Pi_{j},
\end{equation}
where $\Pi_{j}=\sum^{m_{1}+\cdots+ m_{j}}_{i=m_{1}+\cdots+ m_{j-1}+1}\ket{i}\bra{i}$ and the $\mathcal{S}(\sigma)$ by
 \begin{equation}
\mathcal{S}(\sigma)=\{\psi\in \mathcal{L}(\mathcal{K},\mathcal{H}): \psi^{\dagger}\psi=P(\sigma)\}.
\end{equation}
Thus we have constructed a principal fiber bundle
 \begin{equation}\xymatrix{
\mathcal{U}(\sigma)\ar@{^{(}->}[r] & \mathcal{S}(\sigma)\ar[r]^{\pi} &\mathcal{D}(\sigma)},
\end{equation}
 where the gauge group $\mathcal{U}(\sigma)$ is defined by
 \begin{equation}
\mathcal{U}(\sigma)=\{U\in \mathcal{K}: UP(\sigma)=P(\sigma)U\}
\end{equation}
with corresponding Lie algebra $\mathbb{u}(\sigma)$. Note the action of $\mathcal{U}(\sigma)$ on $\mathcal{S}(\sigma)$ is induced by the right action of unitary group $\mathcal{U}(\mathcal{K})$ on $\mathcal{L}(\mathcal{K},\mathcal{H})$.
\begin{example}
If we have one eigenvalue, e.g.,  $\lambda_{1}=1$, then $\mathcal{S}(\sigma)$ is the unit sphere in the Hilbert space $\mathcal{H}$, and $\mathcal{D}(\sigma)$ is the projective space over
$\mathcal{H}$ and our principal fiber bundle $\pi$ is the generalized Hopf bundle (\ref{Hopf}) discussed in pervious section.
\end{example}
In general  $\mathcal{S}(\sigma)$ is diffeomorphic to the Stiefel variety of $k$-frames in $\mathcal{H}$. The following theorem has been proven in \cite{GUR}.
\begin{thm}
Let $\mathbb{u}(\mathcal{K}^{*})$ be the space of all functionals on $\mathbb{u}(\mathcal{K})$ and the momentum mapping
\begin{equation}
J: \mathcal{L}(\mathcal{K},\mathcal{H})\longrightarrow\mathfrak{u}(\mathcal{K}^{*})
\end{equation}
is defined by $J(\psi)=\mu_{\psi^{*}\psi}$ where $\mu_{\hat{A}}(\xi)=i\hbar \mathrm{Tr} (\hat{A}\xi)$, for any hermitian operator $\hat{A}$ on $\mathcal{K}$. Then $J$ is a coadjoint-equivariant map for the Hamiltonian $\mathcal{U}(\mathcal{K})$-action on $\mathcal{L}(\mathcal{K},\mathcal{H})$ and $\mu_{\psi^{*}\psi}$ is a regular value of  $J$ whose isotropy group acts properly, and freely on $J^{-1}(\mu_{\psi^{*}\psi})$.
\end{thm}
Thus we can define $\mathcal{S}(\sigma)=J^{-1}(\mu_{\psi^{*}\psi}) $.  Moreover, $\mathcal{U}(\sigma)$  is the isotropy group of $\mu_{\psi^{*}\psi}=\mu_{P(\sigma)}$ and our principal fiber bundle $\xymatrix{
\mathcal{U}(\sigma)\ar@{^{(}->}[r] & \mathcal{S}(\sigma)\ar[r]^{\pi} &\mathcal{D}(\sigma)}$ is equivalent to the following reduced space submersion
\begin{equation}
J^{-1}(\mu_{P(\sigma)})\longrightarrow J^{-1}(\mu_{P(\sigma)})/\mathcal{U}(\sigma).
\end{equation}
For more information see our recent work \cite{GUR}.

\subsubsection{Riemannian and symplectic structures on $\mathcal{S}(\sigma)$ and $\mathcal{D}(\sigma)$}
Next, we will discuss the Riemannian and symplectic structures on $\mathcal{S}(\sigma)$ and
$\mathcal{D}(\sigma)$. First,  we note that the space $\mathcal{L}(\mathcal{K},\mathcal{H})$ is also equipped with a Hilbert- Schmidt inner product.
 $2\hbar$ times the real part of Hilbert- Schmidt inner product defines  a Riemannian metric
\begin{equation}
G(X,Y)= \hbar\mathrm{Tr}(X^{\dagger}Y+Y^{\dagger}X)
\end{equation}
on $\mathcal{L}(\mathcal{K},\mathcal{H})$ and $2\hbar$  times   the imaginary part of Hilbert- Schmidt inner product defines a sympletic form
\begin{equation}
\Omega(X,Y)=-i\hbar \mathrm{Tr}(X^{\dagger}Y-Y^{\dagger}X)
\end{equation}
on $\mathcal{L}(\mathcal{K},\mathcal{H})$. The unitary  groups $\mathcal{U}(\mathcal{H})$ and $\mathcal{U}(\mathcal{K})$ act on $\mathcal{L}(\mathcal{K},\mathcal{H})$ from left and right respectively by isometric and symplectic transformations $L_{U}(\psi)=U\psi$ and $R_{V}(\psi)=\psi V$.
Moreover, we let $\mathbb{u}(\mathcal{H})$ and $\mathbb{u}(\mathcal{K})$ be the Lie algebras of $\mathcal{U}(\mathcal{H})$ and $\mathcal{U}(\mathcal{K})$ respectively. Furthermore, we define two vector fields $X_{\xi}$ and $X_{\eta}$  corresponding to $\xi$ in $\mathbb{u}(\mathcal{H})$ by
\begin{equation}
X_{\xi}(\psi)=\frac{d}{dt}\left[L_{\exp (t\xi)}(\psi)\right]_{t=0}=\xi \psi
\end{equation}
and  $\eta$ in $\mathbb{u}(\mathcal{K})$ by
\begin{equation}
X_{\eta}(\psi)=\frac{d}{dt}\left[R_{\exp (t\eta)}(\psi)\right]_{t=0}=\psi \eta.
\end{equation}
Now, from Marsden-Weinstein-Meyer symplectic reduction theorem \ref{srt} follows that  $\mathcal{D}(\sigma)$ admits  a symplectic form which is pulled back to $\Omega|_{\mathcal{S}(\sigma)}$.
The following  theorem is also proved in \cite{GUR}.
 \begin{thm}
Consider the principal fiber bundle $\xymatrix{
\mathcal{U}(\sigma)\ar@{^{(}->}[r] & \mathcal{S}(\sigma)\ar[r]^{\pi} &\mathcal{D}(\sigma)}$. Then the  projective space $\mathcal{D}(\sigma)$ admits a unique symplectic form such that $\pi^{*}\omega=\Omega|_{\mathcal{S}(\sigma)}$.
\end{thm}
We will also restrict the metric $G$ to a gauge-invariant metric on $\mathcal{S}(\sigma)$.
 The tangent bundle $\mathcal{T}\mathcal{S}(\sigma)$ of  $\mathcal{S}(\sigma)$ can be decompose as
 \begin{equation}
 \mathcal{T}\mathcal{S}(\sigma)=V\mathcal{S}(\sigma)\oplus H\mathcal{S}(\sigma),
 \end{equation}
 where the vertical bundle $V\mathcal{S}(\sigma)= \mathrm{Ker} d \pi$ and the horizontal bundle $H\mathcal{S}(\sigma)= V\mathcal{S}(\sigma)^{\perp}$, see Figure 4. Note that  $d\pi$ is the differential of $\pi$ and $\perp$ denotes the orthogonal complement with respect to the metric $G$. A vector in $V\mathcal{S}(\sigma)$ is called vertical and a vector in $H\mathcal{S}(\sigma)$ is called the horizontal. We also define a unique metric $g$ on $\mathcal{D}(\sigma)$ which makes the map $\pi$ a Riemannian submersion. This mean that the metric $g$  has property that restriction of $d \pi$ to every fiber of the horizontal bundle is an isometry.
\begin{figure}[t]
\centering
\includegraphics[scale=.65]{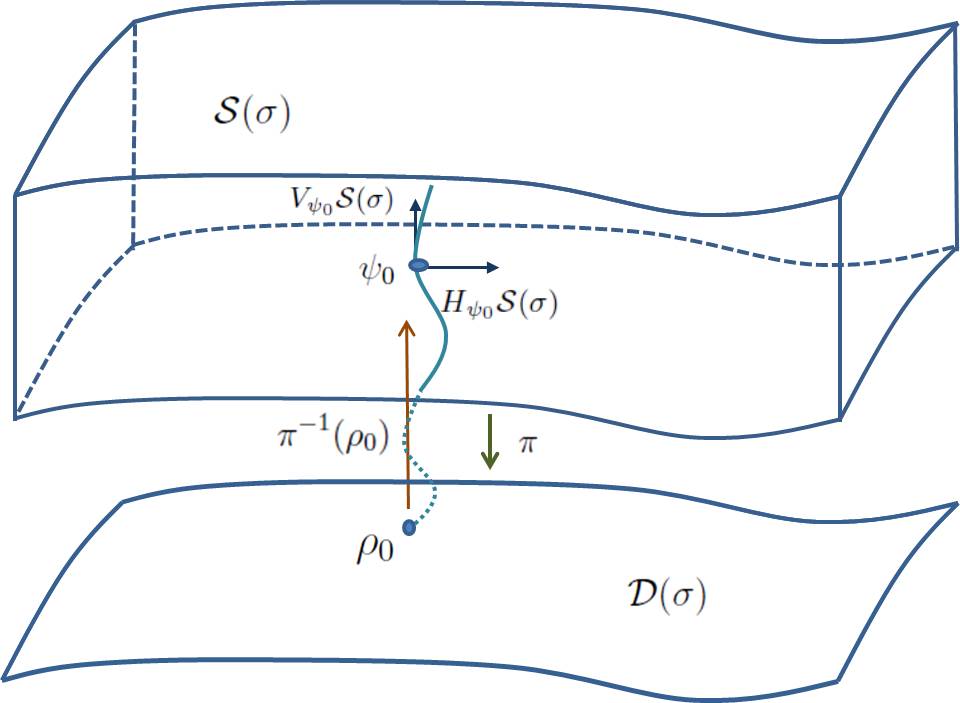}
%
%
\caption{Illustration of the principal fiber bundle $\pi$  and the decomposition of  $\mathcal{T}_{\psi}\mathcal{S}(\sigma)$.}
\label{fig11}      
\end{figure}
Thus we have shown that the total space $\mathcal{S}(\sigma)$  is equipped with a symplectic form and a Riemannian metric. Moreover the quantum phase space  $\mathcal{D}(\sigma)$ is equipped with a symplectic form $\omega$ and a Riemannian metric $g$. We can also show that there exists a compatible  almost complex structure $J$ on $\mathcal{D}(\sigma)$ such that $g(X,Y)=\omega(X,J(Y))$. But to find an explicit expression for $J$ is not an easy task and needs further investigation.

\subsubsection{Mechanical connection }
In this section, we will derive  an explicit  connection on $\mathcal{S}(\sigma)$. The connection is a smooth subbundle $H\mathcal{S}(\sigma)$ of $\mathcal{T}\mathcal{S}(\sigma)$ which is also called an Ehresmann connection.
There is a canonical isomorphism between the Lie algebra $\mathbb{u}(\sigma)$ and the fibers in $V\mathcal{S}(\sigma)$, that is
\begin{equation}
\mathfrak{u}(\sigma)\ni \xi \longmapsto \psi \xi \in V_{\psi}\mathcal{S}(\sigma)
 \end{equation}
  Moreover, $V\mathcal{S}(\sigma)$ is the kernel bundle of gauge invariant mechanical connection $\mathcal{A}:\mathcal{T}\mathcal{S}(\sigma)\longrightarrow \mathfrak{u}(\sigma)$ defined by
\begin{equation}
\mathcal{A}_{\psi}= \mathcal{I}^{-1}_{\psi} \mathcal{J}_{\psi},
 \end{equation}
where $\mathcal{I}: \mathcal{S}(\sigma)\times\mathfrak{u}(\sigma)\longrightarrow\mathfrak{u}(\sigma)^{*}$ is called a locked inertia tensor which  is defined by
\begin{equation}
 \mathcal{I}_{\psi} \xi (\eta)=G (X_{\xi}(\psi), X_{\eta}(\psi))
 \end{equation}
 and
 $\mathcal{J}: \mathcal{T}\mathcal{S}(\sigma)\longrightarrow\mathfrak{u}(\sigma)^{*}$ is called a momentum map which is defined by
\begin{equation}
 \mathcal{J}_{\psi} (X)(\xi)=G (X, X_{\xi}(\psi)).
 \end{equation}
 The locked inertia tensor is of bi-invariant type since $\mathcal{I}_{\psi} $ is an adjoint-invariant form on the Lie algebra $\mathbb{u}(\sigma)$ and it is also independent of the $\psi$. Hence the locked inertia tensor defines a metric on $\mathbb{u}(\sigma)$  as follows
 \begin{equation}\label{metric}
\xi \cdot \eta=\hbar \mathrm{Tr}\left((\xi^{\dagger}\eta+\eta^{\dagger}\xi)P(\sigma)\right).
 \end{equation}
 We will use this metric to derive an explicit expression for the mechanical connection $\mathcal{A}_{\psi}$.
\begin{prop}Let $X\in \mathcal{T}\mathcal{S}(\sigma)$ and $\Pi_{j}=\sum^{m_{1}+\cdots+ m_{j}}_{i=m_{1}+\cdots+ m_{j-1}+1}\ket{i}\bra{i}$. Then
\begin{equation}\label{MC}
\mathcal{A}_{\psi} (X)=\sum^{k}_{j=1} \Pi_{j}\psi^{\dagger}X\Pi_{j}P^{-1}(\sigma).
\end{equation}
 \end{prop}
\begin{proof}
Note that  $\Pi_{j}\psi^{\dagger}X\Pi_{j}P^{-1}(\sigma)\in\mathfrak{u}(\sigma)$. Since
\begin{equation}
\Pi_{j}\psi^{\dagger}X\Pi_{j}P^{-1}(\sigma)=P^{-1}(\sigma)\psi^{\dagger}X\Pi_{j},
\end{equation}
shows that $\Pi_{j}\psi^{\dagger}X\Pi_{j}P^{-1}(\sigma)$ commutes with $P(\sigma)$. Using the definition of $P(\sigma)=\psi^{\dagger}\psi$ one can show that $X^{\dagger}\psi+\psi^{\dagger}X=0$. Thus $\Pi_{j}\psi^{\dagger}X\Pi_{j}P^{-1}(\sigma)$ is anti-hermitian, that is
\begin{equation}
(\Pi_{j}\psi^{\dagger}X\Pi_{j}P^{-1}(\sigma))^{\dagger}
+\Pi_{j}\psi^{\dagger}X\Pi_{j}P^{-1}(\sigma)=
\Pi_{j}\overbrace{(X^{\dagger}\psi+\psi^{\dagger}X)}^{=0}\Pi_{j}P^{-1}(\sigma)=0.
\end{equation}
Finally, we can get an explicit expression for the mechanical connection as follows
\begin{eqnarray}
\sum^{k}_{j=1} \Pi_{j}\psi^{\dagger}X\Pi_{j}P^{-1}(\sigma)\cdot\xi&=&\hbar
\mathrm{Tr}\left(\sum^{k}_{j=1} \Pi_{j}(X^{\dagger}\psi\xi+\xi^{\dagger}\psi^{\dagger}X)\Pi_{j}\right)
\\\nonumber&=&\hbar\mathrm{Tr}\left(X^{\dagger}\psi\xi+\xi^{\dagger}\psi^{\dagger}X\right)
\\\nonumber&=&J_{\psi}(X)(\xi).
\end{eqnarray}
\end{proof}
Next, we will discuss some important applications of the framework in foundations of quantum theory to illustrate the usefulness and  applicability of this formulation of quantum  mechanics.
\subsection{Quantum energy dispersion}\label{qed}
In this section we will consider an important class of observables, namely Hamiltonian $\hat{H}$ operators of the quantum systems. A real-valued  function $H:\mathcal{H}\longrightarrow \mathbb{R}$ of  $\hat{H}$ is called average energy function and it is defined by $H(\rho)=\mathrm{Tr}(\hat{H}\rho)$. If we let $X_{H}$ denotes the Hamiltonian vector field of $H$, then the von Neumann equation governing the dynamics of unitary evolving   density operator can be written as
\begin{equation}
X_{H}(\rho)=\frac{1}{i\hbar}[\hat{H},\rho].
\end{equation}
To prove this we let $X_{H}\in T_{\rho_{0}}D(\sigma)$ with $X_{H}=\dot{\rho}(0)$. Then we want to show that $X_{H}=\frac{1}{i\hbar}[\hat{H}, \rho_{0}]$ for some $\hat{H}$, where $\rho_{0}=\rho(0)$. Now, we consider a curve $\rho(t)$ starting at $\rho_{0}$ with $X_{H}=\dot{\rho}(0)$. Then since $U(\mathcal{H})$ acts transitively on $T\mathcal{D}(\sigma)$, we have
      $\rho(t)=U(t)\rho_{0}U^{\dagger}(t)$ and
      \begin{eqnarray}
        \dot{\rho}(t) &=& \dot{U}(t)\rho_{0}U^{\dagger}(t)+U(t)\rho_{0}\dot{U}^{\dagger}(t)
        \\\nonumber&=& \frac{1}{i \hbar}\hat{H}(t)U(t)\rho_{0}U^{\dagger}(t)-\frac{1}{i \hbar}U(t)\rho_{0}U^{\dagger}(t)\hat{H}(t).
      \end{eqnarray}
      Thus  $\dot{\rho}(0)=\frac{1}{i \hbar}\hat{H}\rho_{0}-\frac{1}{i \hbar}\rho_{0}\hat{H}=\frac{1}{i \hbar}[\hat{H},\rho_{0}]$ and we have shown that $X_{H}=\dot{\rho}(0)=\frac{1}{i \hbar}[\hat{H},\rho_{0}]$.
The Hamiltonian vector field has a gauge-invariant lift $X_{\hat{H}}$ to $\mathcal{S}(\sigma)$ which is defined by
  \begin{equation}
X_{\hat{H}}(\psi)=\frac{1}{i\hbar}\hat{H}\psi.
\end{equation}
The Hamiltonian $\hat{H}$ is said to be parallel at a density operator $\rho$ if $X_{\hat{H}}(\psi)$ horizontal at every $\psi$ in the fiber over $\rho$. Note that $H$ parallel transport $\rho$ if the solution to $\dot{\psi}=X_{\hat{H}}(\psi)$ with initial condition $\psi(0)\in \pi^{-1}(\rho)$ is horizontal. Note also that for any curve $\rho(t)\in \mathcal{D}(\sigma)$ with initial value $\psi_{0}$ in the fiber $\rho(0)$, there is a unique horizontal curve $\psi(t)\in \mathcal{S}(\sigma)$ which is the solution for some Hamiltonian, since the unitary group $\mathcal{U}(\mathcal{H})$ act transitively on $\mathcal{S}(\sigma)$.
If for a known Hamiltonian $\hat{H}$ we define a $\mathbb{u}(\sigma)$-valued field $\xi_{H}$ on $\mathcal{D}(\sigma)$ by
  \begin{equation}
\pi^{*}\xi_{H}=\mathcal{A}\circ X_{\hat{H}},
\end{equation}
then $\xi_{H}\cdot\xi_{H}$ will equal the square of the norm of vertical part of $X_{\hat{H}}$, where the operation $~\cdot~$ defines a metric on $\mathbb{u}(\sigma)$ as in the equation (\ref{metric}).
  \begin{rmk}
The  $\mathbb{u}(\sigma)$-valued  field $\xi_{H}$ is intrinsic to quantum systems.  The complete information about the Hamiltonian $H$ is also included in the field $\xi_{H}$.
\end{rmk}
Next, for a given Hamiltonian, we will establish a relation between the uncertainty function
   \begin{equation}
\Delta H(\rho)=\sqrt{\mathrm{Tr}(\hat{H}^{2}\rho)-\mathrm{Tr}(\hat{H}\rho)^{2}},
\end{equation}
and the intrinsic field $\xi_{H}$.
  \begin{thm}\label{them1}
Let $\xi^{\perp}_{H}$ be the projection of the field $\xi_{H}$ on the orthogonal complement of the unit vector $-i\mathbf{1}\in\mathfrak{u}(\sigma)$.  Then the Hamiltonian vector field $X_{H}$ satisfies
   \begin{equation}
\hbar^{2} g(X_{H}(\rho),X_{H}(\rho))=\Delta H(\rho)^{2}-\xi^{\perp}_{H}(\rho)\cdot \xi^{\perp}_{H}(\rho).
\end{equation}
If the  Hamiltonian $\hat{H}$ is parallel at $\rho$, then $\hbar^{2} g(X_{H}(\rho),X_{H}(\rho))=\Delta \hat{H}(\rho)^{2}$.
\end{thm}
 \begin{proof}
To prove the theorem we start by determining $\mathrm{Tr}(\hat{H}^{2}\rho)$ and $\mathrm{Tr}(\hat{H}\rho)^{2}$ by considering $\psi$  to be a purification of $\rho$.
 Thus we will have
\begin{eqnarray}
\nonumber
\mathrm{Tr}(\hat{H}^{2}\rho) &=& \hbar^{2}G(X_{\hat{H}}(\psi),X_{\hat{H}}(\psi))
\\&=&
\hbar^{2}g(X_{H}(\rho),X_{H}(\rho))+\hbar^{2}\xi_{H}(\rho)\cdot \xi_{H}(\rho)
\end{eqnarray}
 and
 \begin{eqnarray}
\nonumber
\mathrm{Tr}(\hat{H}\rho) &=&i \hbar\mathrm{Tr}(\mathcal{A}_{\psi}(X_{\hat{H}}(\psi))P(\sigma))
\\&=&\nonumber
i\hbar\mathrm{Tr}(\xi_{H}(\rho)P(\sigma))
\\&=& \hbar(-i \mathbf{1})\cdot \xi_{H}(\rho).
\end{eqnarray}
The result follows from
\begin{eqnarray}
\nonumber
\Delta H(\rho)^{2}
&=&\nonumber
\hbar^{2}g(X_{H}(\rho),X_{H}(\rho))+\hbar^{2}\xi_{H}(\rho)\cdot \xi_{H}(\rho)-\hbar^{2}H^{2}
\\&=&
\hbar^{2}g(X_{H}(\rho),X_{H}(\rho))+\hbar^{2}\xi^{\perp}_{H}(\rho)\cdot \xi^{\perp}_{H}(\rho).
\end{eqnarray}
Now, if $ \xi_{H}(\rho)=0$, then we get $\Delta H(\rho)^{2}=\hbar^{2}g(X_{H}(\rho),X_{H}(\rho))$.
\end{proof}
Note that for a pure state the field $\xi^{\perp}_{H}=0$, since the vertical bundle is one- dimensional and so we have $\Delta H^{2}=\hbar^{2}g(X_{H},X_{H})$ which is almost coincides with result given in \cite{Anandan_etal1990}.

\subsection{Quantum distance measure}
In this section, we will consider measuring distance between density operators defined on $\mathcal{D}(\sigma)$ which we have called dynamic distance measure \cite{DD}. The distance of a curves in $\mathcal{D}(\sigma)$ is a geodesic distance and it is defined as the infimum of the lengths of all curves that connect them.
  \begin{thm}
Let $\rho_{0},\rho_{1}\in \mathcal{D}(\sigma)$ be two density operators and $\hat{H}$ be the Hamiltonian of a quantum system. Then distance between $\rho_{0}$ and $\rho_{1}$ is given by
\begin{equation}
\mathrm{Dist}(\rho_{0},\rho_{1})=\inf_{\hat{H}}\frac{1}{\hbar}\int^{t=\tau}_{t=0}
\Delta H(\rho) dt,
\end{equation}
where the infimum is taken over all $\hat{H}$ that solve the following boundary value von Neumann  problem: $\dot{\rho}=X_{H}(\rho)$ with $\rho(0)=\rho_{0}$ and $\rho(1)=\rho_{1}$.

\end{thm}
\begin{proof}
First we note that the length of a curve $\rho=\rho(t)\in\mathcal{D}(\sigma) $ is given by
\begin{equation}
\mathrm{Length}[\rho]=\int^{t=\tau}_{t=0}\sqrt{g(X_{H}(\rho),X_{H}(\rho))}
 dt.
\end{equation}
Now, we will use the result of the theorem \ref{them1}: if $\rho$ is the integral curve of the vector field $X_{H}$, then the length of $\rho$ is given by
\begin{equation} \label{leng}
\mathrm{Length}[\rho]\leq\frac{1}{\hbar}\int^{t=\tau}_{t=0}\Delta H(\rho)
 dt.
\end{equation}
For a Hamiltonian that generates a horizontal lift of $\rho$ we have also equality in  (\ref{leng}) by the theorem \ref{them1}. However, if $\rho$ is a shortest geodesic, then we will have
\begin{equation}
\mathrm{Dist}(\rho_{0},\rho_{1})=\frac{1}{\hbar}\int^{t=\tau}_{t=0}
\Delta H(\rho) dt.
\end{equation}
Thus we have proved the theorem.
\end{proof}
The following theorem is proved in \cite{DD}.
  \begin{thm}
The distance measure $\mathrm{Dist}(\rho_{0},\rho_{1})$ is a proper measure.

\end{thm}
The distance measure $\mathrm{Dist}(\rho_{0},\rho_{1})$ also satisfies the following conditions
\begin{itemize}
  \item \textbf{Positivity:} $\mathrm{Dist}(\rho_{0},\rho_{1})\geq0$.
  \item \textbf{Non-degeneracy:} $\mathrm{Dist}(\rho_{0},\rho_{1})=0$ if and only if $\rho_{0}=\rho_{1}$.
  \item \textbf{Symmetry:} $\mathrm{Dist}(\rho_{0},\rho_{1})=\mathrm{Dist}(\rho_{1},\rho_{0})$.
  \item \textbf{Triangle inequality:} $\mathrm{Dist}(\rho_{0},\rho_{2})\leq \mathrm{Dist}(\rho_{0},\rho_{1})+\mathrm{Dist}(\rho_{1},\rho_{2})$.
  \item \textbf{Unitary invariance:} $\mathrm{Dist}(U\rho_{0}U^{\dagger},U\rho_{1}U^{\dagger})=\mathrm{Dist}(\rho_{0},\rho_{1})$.
\end{itemize}
\begin{example}
Consider a mixed quantum states with $\sigma=(\lambda_{1},\lambda_{2})$ and let $\varepsilon>0$: Then $\psi(t)\in \mathcal{S}(\sigma)$ is given by
\begin{equation}
\psi(t)=\left(
          \begin{array}{cc}
            \sqrt{\lambda_{1} }\cos(\varepsilon t)& \sqrt{\lambda_{2} }\sin(\varepsilon t) \\
             \sqrt{\lambda_{1} }\sin(\varepsilon t)&\sqrt{\lambda_{2} }\cos(\varepsilon t) \\
          \end{array}
        \right)
\end{equation}
for $0\leq t\leq1$. Now, if we set $\rho_{0}=\psi(0)\psi(0)^{\dagger}$ and $\rho_{1}=\psi(1)\psi(1)^{\dagger}$. Then for small $\varepsilon$ we have $\mathrm{Dist}(\rho_{0},\rho_{1})=\mathrm{Length}[\psi]=\varepsilon$.
To be able to compare $\mathrm{Dist}(\rho_{0},\rho_{1})$ with other well-known distance measure we will consider
an explicit formula for Bures distance for density operators on finite dimensional Hilbert space \cite{Dittmann1999}. In particular, the Bures distance on $\mathbb{C}^{2}$ is given by
\begin{equation}
D_{B}(\rho,\rho+\delta\rho)=\frac{1}{4}\mathrm{Tr}\left(\delta\rho\delta\rho+\frac{1}{\det\rho}(\delta\rho-
\rho\delta\rho)\delta\rho^{2}\right)
\end{equation}
An explicit expression for $\mathrm{D}_{B}(\rho_{0},\rho_{1})$ can be found in \cite{Dittmann1993}:
\begin{equation}
\mathrm{D}_{B}(\rho_{0},\rho_{1})=\frac{\lambda_{1}-\lambda_{2}}{\sqrt{2}}|\sin \varepsilon|\sqrt{2+\frac{(\lambda_{1}-
\lambda_{2})^{2}}{2\lambda_{1}\lambda_{2}}\sin^{2}\varepsilon}.
\end{equation}
\end{example}
The reader can find further information on the distance measure in our recent work on the subject \cite{DD}.
\\
A curve in $\mathcal{D}(\sigma)$ is a geodesic if and only
 if its horizontal lifts are geodesics in $\mathcal{S}(\sigma)$, and that the
  distance between two operators in $\mathcal{D}(\sigma)$ equals the length of
   the shortest geodesic that connects the fibers of $\pi$ over the two operators \cite{QSL}.
Let $\rho_{0},\rho\in \mathcal{D}(\sigma)$ with corresponding fibers $\ket{\psi_{0}}\in \pi^{-1}(\rho_{0})$ and $\ket{\psi}\in\pi^{-1}(\rho)$ defined on $ \mathcal{S}(\sigma)$. Moreover, let $\rho_{0}$ and  $\rho$  correspond to the points $p_{0}$ and $p$. Then one can consider a function $\delta_{p_{0}}:\mathcal{D}(\sigma)\longrightarrow \mathbb{R}_{+}$  defined by
 \begin{equation}\delta_{p_{0}}(p)=\mathrm{Tr}(\rho_{0}\rho)
\end{equation}
 called the quantum probability distribution on $\mathcal{D}(\sigma)$. The relation between the quantum probability distribution and the distance measure on the quantum phase space  $\mathcal{D}(\sigma)$ needs further investigation.

\subsection{Geometric uncertainty relation}
In this section, we discuss a geometric uncertainty relation for mixed quantum states \cite{GUR}.
Let $\hat{A}$ be a general  observable on the Hilbert space.  Then  an uncertainty function for  $\hat{A}$   is given by
   \begin{equation}
\Delta A(\rho)=\sqrt{\mathrm{Tr}(\hat{A}^{2}\rho)-\mathrm{Tr}(\hat{A}\rho)^{2}}.
\end{equation}
\begin{rmk}
Note that  almost all theory that we have discussed in section  \ref{qed} about $\mathbb{u}(\sigma)$-valued field $\xi_{H}$ can be applied here by  replacing the Hamiltonian $\hat{H}$ by $\hat{A}$.
\end{rmk}
Now, let  $\hat{A}$ and $\hat{B}$ be two observables. Moreover, let $(A,B)$ and $[A,B]$ be the expectation value functions of $(\hat{A},\hat{B})=\frac{1}{2}(\hat{A}\hat{B}+\hat{B}\hat{A})$ and $[\hat{A},\hat{B}]=\frac{1}{2i}(\hat{A}\hat{B}-\hat{B}\hat{A})$ respectively.  Then the Robertson-Schr\"{o}dinger uncertainty relation \cite{Robertson_1929} is given by
   \begin{equation}
\Delta A(\rho)\Delta B(\rho)\geq \sqrt{((A,B)-AB)^{2}+[A,B]^{2}},
\end{equation}
Next we want to derive a geometric uncertainty relation for mixed quantum states that involves Riemannian metric and symplectic form as we have derived for pure states.
   \begin{thm}
   Let $A$ and $B$ be two observables on the Hilbert space $H$. Then a geometric uncertainty relation for mixed quantum states is given by
   \begin{equation}\label{gur}
\Delta A(\rho)\Delta B(\rho)\geq \frac{\hbar}{2}\sqrt{\{A,B\}^{2}_{g}+\{A,B\}^{2}_{\omega}},
\end{equation}
where $\{A,B\}_{g}=g(X_{A},X_{B})$ is the Riemannian bracket and $\{A,B\}_{\omega}=\omega(X_{A},X_{B})$ is the Poisson bracket of $A$ and $B$.
\end{thm}
   \begin{proof}
First we calculate the expectation value of $\hat{A}$:
\begin{eqnarray}
\nonumber
  A(\rho) &=& i\hbar \mathrm{Tr} (A_{\psi}(X_{\hat{A}}(\psi))P(\sigma))\\\nonumber&=&
  i\hbar \mathrm{Tr} (\xi_{A}(\rho)P(\sigma))
  \\&=& \frac{\hbar}{2}\chi\cdot\xi_{A}(\rho),
\end{eqnarray}
where $\chi=\frac{1}{i\sqrt{2\hbar}}\mathbf{1}$ is the unit vector in the Lie algebra $\mathbb{u}(\sigma)$. Thus the expectation value function of $\hat{A}$ is proportional to length of the projection of $\xi_{A}$ on $\chi$. Similarly for the observable $\hat{B}$ we have
\begin{eqnarray}
  B(\rho) &=&  \frac{\hbar}{2}\chi\cdot\xi_{B}(\rho).
\end{eqnarray}
We also need to estimate $(A,B)(\rho)$:
\begin{eqnarray}
  (A,B)(\rho) &=& \frac{\hbar}{2} G(X_{\hat{A}}(\psi),X_{\hat{B}}(\psi))\\\nonumber&=&
 \frac{\hbar}{2} g(X_{A}(\rho),X_{B}(\rho))+\frac{\hbar}{2}\xi_{A}(\rho)\cdot\xi_{B}(\rho),
\end{eqnarray}
and $[A,B](\rho)$:
\begin{eqnarray}
  [A,B](\rho) &=& \frac{\hbar}{2} \Omega(X_{\hat{A}}(\psi),X_{\hat{B}}(\psi))\\\nonumber&=&
 \frac{\hbar}{2} \omega(X_{A}(\rho),X_{B}(\rho)).
\end{eqnarray}
Let  $X^{\perp}_{A}(\rho)$ and $X^{\perp}_{B}(\rho)$ be the projection of $X_{A}(\rho)$ and $X_{B}(\rho)$ into the orthogonal  complement of the unit vector $\chi$ in $\mathbb{u}(\sigma)$. Then we have
\begin{eqnarray}
  (A,B)(\rho) - A(\rho)B(\rho)&=& \frac{\hbar}{2} \{A(\rho),B(\rho)\}_{g}+\frac{\hbar}{2}\xi^{\perp}_{A}(\rho)\cdot\xi^{\perp}_{B}(\rho)
\end{eqnarray}
and in particular for the observable $\hat{A}$ we get
\begin{eqnarray}\label{est}
 \Delta A^{2}= (A,A)(\rho) - A(\rho)A(\rho)\geq  \frac{\hbar}{2} \{A(\rho),A(\rho)\}_{g}.
\end{eqnarray}
Now, we let $X^{\|}_{A}$ and $X^{\|}_{B}$ denote the horizontal components of vector fields  $X_{A}$ and $X_{B}$ respectively. Then we have
\begin{eqnarray} \label{ineq}
\{A,A\}_{g}\{B,B\}_{g}
&=&G(X^{\|}_{\hat{A}},X^{\|}_{\hat{A}})G(X^{\|}_{\hat{B}},X^{\|}_{\hat{B}})\\\nonumber&\geq&
G(X^{\|}_{\hat{A}},X^{\|}_{\hat{B}})^{2}+\Omega(X^{\|}_{\hat{A}},X^{\|}_{\hat{B}})^{2}
\\\nonumber&=&\{A,B\}^{2}_{g}+\{A,B\}^{2}_{\omega},
\end{eqnarray}
where we have applied Cauchy-Schwarz inequality. Combining equations  (\ref{est}) and (\ref{ineq}) we get the geometric uncertainty relation given by equation (\ref{gur}).
\end{proof}
\begin{example}
For a mixed quantum state with $\sigma=(\lambda_{1},\lambda_{2})$ defined on $\mathbb{C}^{2}$, the geometric uncertainty relation for observables $S_{x}$ and $S_{y}$ is given by \cite{gur}
                              \begin{equation}\label{gur1}
\Delta S_{x}(\rho)\Delta S_{y}(\rho)\geq \frac{\hbar}{2}(\lambda_{1}-\lambda_{2}).
\end{equation}
\end{example}
For a detail comparison between geometric uncertainty and the Robertson-Schr\"{o}dinger uncertainty relation see \cite{GUR}.
\subsection{Geometric postulates of quantum mechanics for general mixed states}
We have introduced   a geometric formulation  for  mixed quantum states. We have shown that  $\mathcal{D}(\sigma)$ is a symplectic manifold equipped with a symplectic form $\omega $ and a Riemannian metric $g$. We can also show that there is an almost complex structure $J$ on $\mathcal{D}(\sigma)$ which is compatible with $\omega$ and
$g$. But we are not able to find an explicit expression for $J$.  We leave this question for further investigation and  we will write down  a set of  postulates which are a direct generalization of postulates for pure quantum states.
 \begin{itemize}
   \item \textbf{Physical state:} There is a one-one correspondent between points of the projective Hilbert space $ \mathcal{D}(\sigma)$, which is a symplectic manifold equipped with a symplectic form $\omega $ and a Riemannian metric $g$, and the physical states of mixed quantum systems.
 \item \textbf{Observables:} Let $f:\mathcal{D}(\sigma)\longrightarrow \mathbb{R}$ be a real-valued, smooth function on $\mathcal{D}(\sigma)$ which preserves  the symplectic form $\omega $ and the  Riemannian metric $g$. Then the observables or measurable physical quantity is presented by $f$.
   \item \textbf{Quantum evolution:} The evolution of closed mixed quantum systems is determined by the flow on $ \mathcal{D}(\sigma)$, which preserves the symplectic form $\omega $ and the  Riemannian metric $g$. Since we considering finite-dimensional cases, the flow is given by integrating Hamiltonian vector field $X_{H}$ of the observable $\hat{H}$.
 \end{itemize}
  \begin{rmk}
The measurement postulate needs  further investigation. In particular, we need to define a general metric  that is valid on different orbits.
\end{rmk}
\begin{rmk}
 A weakness in the above geometric postulate of quantum mechanics for mixed states is that we are not able to show that the quantum phase space $\mathcal{D}(\sigma)$ is  a K\"{a}hler manifold. But there is another geometric formulation of quantum mechanics  that is based on K\"{a}hler structures  \cite{Hosh}, where the quantum phase space is actually a K\"{a}hler manifold. Thus if one wants to make sure that the quantum phase space is K\"{a}hler manifold, then it would be a better choice to write down the geometric postulate of quantum mechanics based on that geometric framework.
\end{rmk}
To summarize, we have discussed the geometric postulate of quantum mechanics for mixed quantum states for the sake of completeness and this topic still needs further investigation.

\subsection{Geometric phase for mixed quantum states}
We have discussed a fiber bundle approach to geometric phase of pure quantum states in section \ref{GP}. In this section we will extended the discuss to mixed quantum states.
Uhlmann \cite{Uhlmann1986, Uhlmann(1991)} was among the first to
develop a theory for geometric phase of  mixed quantum states based on purification.
Another approach to geometric phase for mixed  quantum states was proposed in \cite{Sjoqvist} based on quantum interferometry. Recently, we have introduced an operational geometric phase for mixed quantum states, based on holonomies \cite{GP}. Our geometric phase generalizes the standard definition of geometric phase for mixed states, which is based on quantum interferometry and it is rigorous, geometrically elegant, and applies to general unitary evolutions of both non\-degenerate and degenerate mixed states. Here we give a short introduction to such a geometric phase for mixed quantum states.
Let $\rho$ be a curve in $\mathcal{D}(\sigma)$. Then the horizontal lifts of $\rho$ defines a parallel transport operator $\Pi[\rho]$ from the fibre over $\rho(0)$ onto the fibre $\rho(\tau)$ as follows
\begin{equation}
\Pi[\rho]\Psi_{0}=\Psi_{\|}(\tau),
\end{equation}
where $\psi_{\|}(\tau)$ is horizontal lift of $\rho$ extending from $\psi_{0}$ defined by
\begin{equation}\label{horiz}
\psi_{||}(t)=\psi(t)\exp_{+}\left(-\int_{0}^{t}\mathcal{A}_\psi(\dot \psi)dt\right),
\end{equation}
where $\exp_{+}$ is the positive time-ordered exponential and $\mathcal{A}$ is the mechanical connection defined by equation (\ref{MC}), see Figure 5.
\begin{figure}[t]
\centering
\includegraphics[scale=.65]{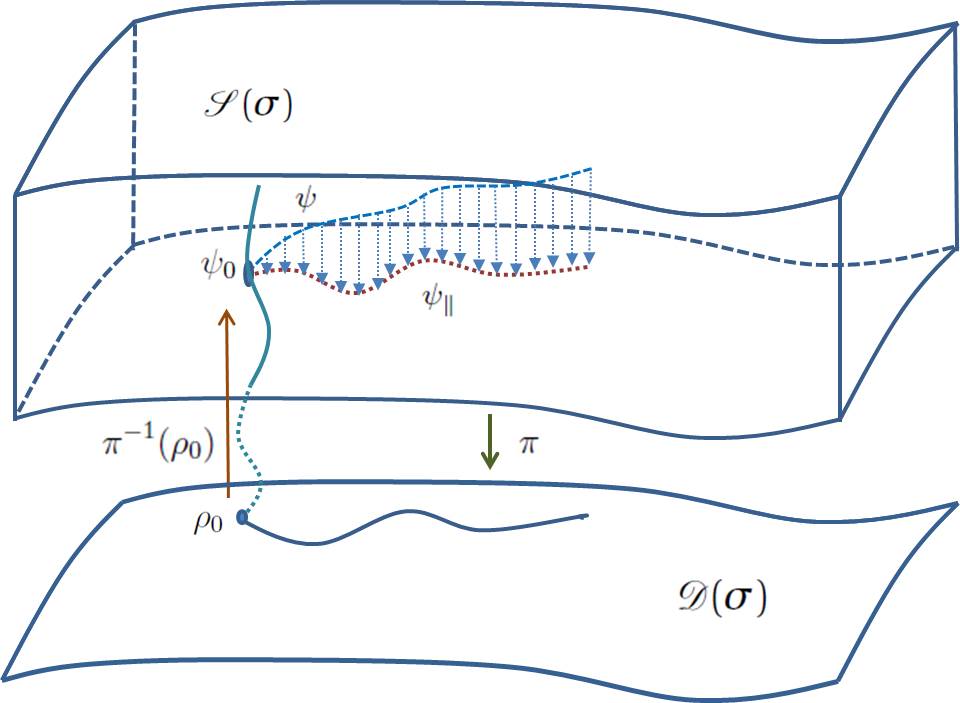}
%
%
\caption{Illustration of horizontal lift $\psi_{\|}(t)$.}
\label{fig11}      
\end{figure}
The geometric phase of $\rho$ is defined by
\begin{eqnarray}
\gamma_{g}(\rho)&=&\arg \mathrm{Tr}(P(\sigma)\mathrm{Hol}(\rho))
\\\nonumber&=&
\arg \mathrm{Tr}(\psi^{\dagger}_{0}\Pi[\rho]\psi_{0})
\\\nonumber&=&\arg \mathrm{Tr}(\psi^{\dagger}_{\|}(0)\psi_{||}(\tau))
\end{eqnarray}
where $\mathrm{Hol}(\rho)$ is the holonomy of $\rho$.
\begin{example}
Consider a mixed qubit state represented by
$\rho=\frac{1}{2}\left(
       \begin{array}{cc}
         1+\cos\vartheta & p\sin\vartheta \\
         p\sin\vartheta & 1-\cos\vartheta \\
       \end{array}
     \right)
$ with $0\leq \vartheta< 2\pi$ and a unitary operator $U(t)=\mathrm{diag}(e^{- it}, e^{ it})$ with $0\leq t \leq 2\pi$ . Then the geometric phase is given by \cite{GP2}
\begin{eqnarray}
\gamma_{g}(\rho)&=&\arg \mathrm{Tr}(\psi^{\dagger}_{\|}(0)\psi_{||}(\tau))
\\\nonumber&=&\arg \left(-\frac{1}{2}(1+p) e^{i \pi \cos \vartheta}-\frac{1}{2}(1-p) e^{-i \pi \cos \vartheta}\right).
\end{eqnarray}
\end{example}
Since we have access to all elements of the holonomy group of $\rho$, we are also able to defined higher order geometric phases for mixed quantum states. We will not discuss higher geometric phases here and refer the interested reader to \cite{GP}.

\section{Conclusion}\label{sec4}
In this work, we have given a concise introduction to geometric formulation of  quantum mechanics based on principal fiber bundle and momentum map. We divided our presentation in three parts. In the first part we have given an  introduction to Hamiltonian dynamics, principal fiber bundle, and momentum map.
In the second part of the text we  have discussed geometry of pure quantum systems including geometric characterization of quantum phase space, quantum dynamics, geometric phase, and quantum measurement  of pure states. We also have discussed some applications of geometric quantum mechanics of pure states such as geometric uncertainty relation and reviewed the geometric postulates of quantum mechanics. In the third  part of the text we have considered the geometric formulation of general quantum states represented by density operators. Our presentation was mostly based on our recent geometric formulation of mixed quantum states. After  a short introduction to the idea of the framework we moved to discuss the applications. We have discussed the quantum energy dispersion, geometric phase, and geometric uncertainty relation for mixed quantum states. We have also tried to extend the geometric postulates of quantum mechanics into mixed quantum states. But this topic definitely needs further investigation.\\
 The results  we have reviewed and discussed in this work give a  very interesting insight on geometrical structures of  quantum systems and on our understanding of geometrical nature  of quantum theory. We are also convinced  that geometric formulation of quantum theory will have an impact on our understanding of physical reality. The geometric framework will also provide us with many applications  waiting to be discovered. We hope that this work could encourage  reader to contributed to this exiting field of research.

\begin{flushleft}
\textbf{Acknowledgments:} The author acknowledges useful comments and discussions with O. Andersson and Professor R. Roknizadeh.
\end{flushleft}

\end{document}